\newif\ifArxiv
\Arxivfalse

\ifArxiv
\documentclass[11pt]{article}
\else
\documentclass[a4paper,USenglish,autoref,thm-restate,numberwithinsect]{lipics-v2021}

\nolinenumbers
\hideLIPIcs 
\fi

\newif\ifMoreBalanced
\MoreBalancedtrue

\newif\ifTableOfContents
\TableOfContentstrue

\newif\ifClearPageAfterAbstract
\ClearPageAfterAbstractfalse

\newif\ifClearPageAfterTableOfContents
\ClearPageAfterTableOfContentstrue

\newif\ifLineNumbers
\LineNumbersfalse

\newif\ifCompileDate
\CompileDatefalse

\newif\ifFrontPageMessage
\FrontPageMessagefalse
\newcommand{\FrontPageMessage}{\bfseries Preliminary draft, not for distribution}

\newif\ifComments
\Commentstrue

\newif\ifSuppressCustomWarnings 
\SuppressCustomWarningsfalse

\newcommand{\TitleMetadataForHyperref}{Fast Deterministic Distributed Degree Splitting}
\newcommand{\AuthorMetadataForHyperref}{Yannic Maus, Alexandre Nolin, Florian Schager}

\usepackage{1a_most_macros}

\usepackage{1b_editorial}

\title{Fast Deterministic Distributed Degree Splitting}
\author{Yannic Maus}{TU Graz, Graz, Austria}{yannic.maus@tugraz.at}{https://orcid.org/0000-0003-4062-6991}{}
\author{Alexandre Nolin}{Télécom SudParis, Institut Polytechnique de Paris, Palaiseau, France}{alexandre.nolin@telecom-sudparis.eu}{https://orcid.org/0000-0002-3952-0586}{Funded by the European Union. Views and opinions expressed are however those of the author(s) only and do not necessarily reflect those of the European Union or the European Research Council. Neither the European Union nor the granting authority can be held responsible for them. This work is supported by ERC grant \href{https://doi.org/10.3030/101162747}{OLA-TOPSENS} (grant agreement number 101162747) under the Horizon Europe funding programme.}
\author{Florian Schager}{TU Graz, Graz, Austria}{florian.schager@tugraz.at}{https://orcid.org/0009-0009-3923-051X}{}

\authorrunning{Y. Maus, A. Nolin, and F. Schager}

\Copyright{Yannic Maus, Alexandre Nolin, and Florian Schager}



\funding{This research was funded in whole or in part by the Austrian Science Fund (FWF) \url{https://doi.org/10.55776/I6915}. For open access purposes, the author has applied a CC BY public copyright license to any author-accepted manuscript version arising from this submission. 
}

\newcommand{\lovasz}{Lov\'{a}sz\xspace}

\begin{document}

\newcommand{\makeAbstract}{
\begin{abstract}
    We obtain faster distributed algorithms for the directed and undirected degree splitting problems.
    In the \textit{directed degree splitting} problem the goal is to orient the edges of a graph to minimize the discrepancy between in-degree and out-degree at every vertex.
    In the \textit{undirected degree splitting} problem the goal is to color the edges red and blue to minimize the discrepancy between red-degree and blue-degree at every vertex.

    In particular, we design a deterministic \LOCAL algorithm of complexity $\bigO(\varepsilon^{-1} \cdot \log n)$ for computing a directed degree split with discrepancy at most $\varepsilon \cdot \deg(v)$ for every vertex $v \in V$.
    The dependence on $n$ in our runtime is optimal due to the $\Omega(\log n)$ lower bound by Chang, Kopelowitz and Pettie (STOC 2016), which already holds for $\varepsilon = 1/3$.
    We improve upon the previous result by Ghaffari, Hirvonen, Kuhn, Maus, Suomela and Uitto (DISC 2017) by eliminating all additional dependence on $\varepsilon$ beyond the leading $\varepsilon^{-1}$ factor.
    Furthermore, we extend our algorithm to compute undirected degree splits with the same discrepancy guarantee and in the same asymptotic runtime.
    Beyond our deterministic \LOCAL results, we also provide faster randomized algorithms for the same problems.

    A key ingredient in our approach is a connection to the hypergraph sinkless orientation (HSO) problem introduced by Brandt, Maus, Narayanan, Schager and Uitto (SODA 2025).
    We extend their randomized HSO algorithm to apply to a larger class of instances, which in turn reduces the runtime of our randomized degree splitting algorithms. We also provide deterministic and randomized \CONGEST algorithms for solving HSO.

    Degree splittings are a natural tool to solve other graph problems via a recursive \textit{divide-and-conquer} strategy.
    As a demonstration of this claim, we apply our degree splitting algorithm to show that $(3 / 2 + \varepsilon)\Delta$-edge coloring can now be solved in $\bigO(\varepsilon^{-1} \cdot \log^2 \Delta \cdot \log n + \varepsilon^{-2} \cdot \log n)$ rounds in \LOCAL.
    For constant $\varepsilon$ and $\Delta = \bigO(2^{\log^{1/3} n})$, this runtime matches the current state-of-the-art for $(2\Delta - 1)$-edge coloring in Ghaffari and Kuhn (FOCS 2021).
\end{abstract}
}

\maketitle
\makeAbstract
\clearpage

\ifCompileDate
\begin{center}
    \hfill\today
\end{center}
\fi

\ifFrontPageMessage
\begin{center}
    \FrontPageMessage
\end{center}
\fi

\ifClearPageAfterTableOfContents
\pagenumbering{roman}
\thispagestyle{empty}
\fi
\ifClearPageAfterAbstract
\pagenumbering{roman}
\thispagestyle{empty}
\clearpage
\else
\ifTableOfContents
\bigskip
\fi
\fi

\ifTableOfContents
\ifClearPageAfterTableOfContents
\thispagestyle{empty}
\fi
\tableofcontents
\ifClearPageAfterTableOfContents
\clearpage
\fi
\fi

\pagenumbering{arabic}

\section{Introduction}

In this work, we present improved distributed algorithms for degree splitting problems in the LOCAL model. 
Degree splitting is a fundamental primitive in distributed graph algorithms and we showcase its usefulness by providing faster algorithms for classic edge coloring problems. 

\subparagraph{Model of computation \cite{linial92}.}
We work in the standard LOCAL model of distributed computing. The network is modeled as a graph, where each node hosts a processor and communication proceeds in synchronous rounds. In each round, nodes can exchange arbitrarily large messages with their neighbors and perform unbounded local computation. The complexity of an algorithm is measured by the number of communication rounds.
The related \CONGEST model additionally imposes a $\mathcal{O}(\log n)$ restriction on the size of each message.

\subparagraph{Degree splitting.} Degree splitting problems are among the most basic problems in distributed graph algorithms that resist straightforward greedy solutions. Despite their simple formulation, they have received considerable attention due to both their intrinsic combinatorial interest and their applications, e.g., \cite{BFHKLRSU_stoc16,GS_soda17,weakSplitting19,GHKMSU_dc20,BKKLOPPRSSSU_sinkless_made_simple_sosa23,HMN22,Davies23,BKP_25,BMNSU_soda25}. In the \emph{undirected degree splitting} problem, the task is to color the edges of a graph with two colors, red and blue, such that the discrepancy between the number of incident red and blue edges at every node is minimized. The objective in the  \emph{directed degree splitting} problem is to orient the edges of a graph such that for each node the discrepancy between its in- and out-degree is minimized.

\subparagraph{Which discrepancy is possible?} From a centralized perspective, the problem admits a simple solution  via an Euler tour argument, and more generally it connects to classical discrepancy theory~\cite{chazelle2000discrepancy}. By adding a dummy vertex, all odd-degree nodes can be paired up to obtain an Eulerian graph.
Alternating the colors along an Euler tour and then removing the dummy vertex yields a discrepancy of at most $2$. In the directed variant, orienting edges along such a tour even gives a discrepancy of at most $1$.
In the distributed setting, however, the problem becomes substantially more challenging. The global structure exploited by standard centralized arguments is not directly accessible, and known distributed approaches rely on more involved techniques such as augmenting paths \cite{GS_soda17} or bootstrapping techniques~\cite{GHKMSU_dc20} for amplifying balanced orientations. 

\subparagraph{Degree splittings in the bigger scope of distributed complexity theory.}
Although degree splitting problems are easy to state and well understood centrally, they capture a fundamental source of difficulty in distributed computation. The necessity of global coordination inherent in these problems can be formalized: each admits a logarithmic lower bound on the round complexity, which holds even for constant-degree graphs~\cite{CKP_siamcomp19,weakSplitting19}.
In contrast to the classical \emph{'big four'} problems---maximal independent set, maximal matching, $(\Delta+1)$-vertex coloring, and $(2\Delta-1)$-edge coloring---which admit $\bigO(\log^* n)$-round algorithms on constant-degree graphs and whose partial solutions can always be locally extended to full ones \cite{linial92,PR01}, degree splittings fall into a qualitatively different regime.
They are closely tied to the \lovasz\ Local Lemma (LLL) \cite{CP19,weakSplitting19,HMN22,Davies23,BKP_25}, which captures the task of simultaneously avoiding a family of locally dependent bad events---for instance, that a node’s discrepancy exceeds a prescribed threshold. The distributed complexity of problems in the LLL regime is central to the landscape of locally checkable labeling problems and underlies many recent advances in sublogarithmic and randomized distributed algorithms \cite{GKM17,BHKLOS18,GHK18,CKP_siamcomp19,CP19,Chang20,BBOS20,BBOS21,GRB22}. At the same time, obtaining a sharp understanding of the distributed complexity of LLL-type problems remains a major open direction in the field.

\subparagraph{Degree splitting and the logarithmic time complexity barrier.}
The logarithmic barrier already appears in the \emph{sinkless orientation} problem, where the goal is to compute an orientation in which each node of degree $3$ or more has at least one outgoing edge.
The problem is equivalent to computing a directed degree splitting with discrepancy $\max(\deg(v)-1,2)$.
The problem is known to have a $\Theta(\log n)$ round complexity on general graphs \cite{CKP_siamcomp19,GS_soda17}, making it one of the few natural problems with a tight intermediate complexity on general graphs. 
The same logarithmic lower bound extends to weak (undirected) degree splitting~\cite{weakSplitting19} where the goal is to color  edges red and blue so that each node has at least one edge of each color. 

Thus, the $\Theta(\log n)$ barrier already arises for the largest meaningful discrepancy. Consequently, the goal is not to circumvent this barrier but to understand the trade-offs within it, namely how the round complexity depends on the allowed discrepancy. This dependence is crucial in recursive divide-and-conquer applications, where one requires discrepancy $\varepsilon \cdot \deg(v)$ with $\varepsilon=o(1)$ and the overall efficiency hinges on how the runtime scales with $\varepsilon$.

\subsection{Our Contributions}
Our main contribution is a faster and conceptually simpler algorithm for computing both directed and undirected degree splits. 

\begin{restatable}{theorem}{SplittingThm}
    \label{thm:balanced-local}
    For every $\varepsilon > 0$, there is a deterministic \LOCAL algorithm with round complexity $\bigO(\varepsilon^{-1} \cdot \log n)$ that computes a directed degree splitting of any $n$-vertex graph such that the discrepancy at each vertex $v$ of degree $\deg(v)$ is at most $\varepsilon \cdot \deg(v) + 1$ if $\deg(v)$ is odd and at most $\varepsilon \cdot \deg(v) + 2$ if $\deg(v)$ is even.

    The same statement holds for the undirected degree splitting problem.
\end{restatable}

Our algorithm is conceptually simpler than the augmenting path-based approach of \cite{GS_soda17} and the bootstrapping framework of \cite{GHKMSU_dc20}. As a consequence, we obtain a clean runtime bound of $\bigO(\varepsilon^{-1} \cdot \log n)$, improving upon the $\bigO(\varepsilon^{-1} \log \varepsilon^{-1} (\log \log \varepsilon^{-1})^{1.71} \cdot \log n)$ bound of \cite{GHKMSU_dc20}. For $\eps=1/\Delta$ our algorithm runs in $\bigO(\Delta \cdot \log n)$ time and obtains an almost optimal discrepancy. Lowering the discrepancy further cannot be expected from an efficient algorithm, as computing a perfect split for even-degree nodes already incurs an $\Omega(n)$ lower bound via a reduction from $2$-coloring on cycles.

\begin{restatable}{theorem}{SplittingThmRand}
    \label{thm:balanced-local-rand}
    For every $\varepsilon > 0$, there is a randomized \LOCAL algorithm with round complexity 
    ${\bigO(\varepsilon^{-1} \cdot 10^{\log^* \varepsilon^{-1}} \cdot \log \log n)}$ that computes  a directed degree splitting of any $n$-vertex graph such that the discrepancy at each vertex $v$ of degree $\deg(v)$ is at most $\varepsilon \cdot \deg(v) + 1$ if $\deg(v)$ is odd and at most $\varepsilon \cdot \deg(v) + 2$ if $\deg(v)$ is even.

    The same statement holds for the undirected degree splitting problem.
\end{restatable}

Note that the factor $10^{\log^{\ast} \varepsilon^{-1}}$ is a very slowly growing function and is $o(\log^{(k)} \varepsilon^{-1})$ for every constant $k > 0$.

\subparagraph{Randomized Hypergraph Sinkless Orientation.}
The runtime of our randomized algorithm directly depends on the state-of-the-art for solving instances of the \emph{Hypergraph Sinkless Orientation} (HSO) problem, which we use as a subroutine.
Our algorithm would have a higher complexity of 
$\bigO(\varepsilon^{-1} \log \varepsilon^{-1} \cdot \log \log n)$
if using the published state-of-the-art.
Therefore, another contribution of our work is to extend the set of instances of Hypergraph Sinkless Orientation that are known to be solvable in $\bigO(\log \log n)$ rounds by a randomized algorithm.

\begin{restatable}{theorem}{betterHSOthm}
    \label{thm:better-rand-HSO}
    There is a randomized distributed algorithm that computes a hypergraph sinkless orientation in any hypergraph $H$ with maximum rank $r$ and minimum degree $\delta \geq 2 x \cdot 10^{\log^{\ast} r} \cdot r$
    in $\bigO(\log_x \log n)$ rounds of \LOCAL for any $x \geq 2$.
\end{restatable}

\subparagraph{Implications for edge coloring.}

The improved dependence on the discrepancy is particularly useful in applications with a recursive divide-and-conquer structure. As a corollary, we obtain the following result on edge coloring with few colors. The algorithm repeatedly applies \autoref{thm:balanced-local} to partition the graph into subgraphs of progressively smaller maximum degree, assigning disjoint color palettes to each part. Once the degrees become sufficiently small, a separate procedure is used to complete the coloring.

\begin{restatable}{corollary}{corEdgeColoring}
For every $\varepsilon > 0$, there is a deterministic \LOCAL algorithm with round complexity $\bigO(\eps^{-1} \cdot \log^2\Delta \cdot \log n + \varepsilon^{-2} \cdot \log n)$ to edge color any (multi)-graph with maximum degree $\Delta$ with $(3/2+\eps)\Delta$ colors. 
\end{restatable}

Note that $3\Delta/2$ colors is a natural threshold for coloring multi-graphs---according to Shannon's Theorem \cite{Shannon49} this is an upper bound on the chromatic index of any multigraph. This bound is tight as witnessed by the family of Shannon multigraphs \cite{Vizing65}.
Our result improves upon the previous state-of-the-art of $\bigO(\varepsilon^{-2} \cdot \log^{2}\Delta \cdot \poly\log\log\Delta \cdot \log n)$ rounds in \cite{BMNSU_soda25}.
Note that for constant $\varepsilon$ and $\Delta = \bigO(2^{\log^{1/3} n})$ it even matches the state-of-the-art for the greedily solvable $(2\Delta -1)$-edge coloring problem \cite{GK21}.

\subparagraph{Limitations for Degree Splitting Algorithms.}
We have already discussed the inherent logarithmic barrier for degree splitting problems. Our algorithm achieves a runtime of $\bigO(\varepsilon^{-1} \cdot \log n)$ for a discrepancy of $\varepsilon \cdot \deg(v)$ at every vertex $v \in V$. 
The first distributed algorithm for degree splitting \cite{GS_soda17} is based on augmenting paths. We prove that the $\bigO(\varepsilon^{-1} \cdot \log n)$ bound is tight for naïve augmenting-path–based approaches, even though our own algorithm does not rely on augmenting paths.
Formally, we show a \emph{mending radius} lower bound. This quantity measures, when extending a given partial solution to a new element of the graph, whether the extension can always be performed without changing the partial solution far from where it is being extended~\cite{BHMORS_sirocco22}.

\begin{restatable}{theorem}{thmMendability}
Computing a directed degree splitting with discrepancy at most $\eps \cdot \Delta$ at each vertex has mending radius $\Omega(\eps^{-1} \cdot \log (n/\Delta))$.

    The same statement holds for the undirected degree splitting problem.
\end{restatable}

\subparagraph{Open Problem I: Optimality of the runtime.}
Our algorithms achieve a round complexity of $\bigO(\varepsilon^{-1}\log n)$ for a local discrepancy of $\varepsilon \cdot \deg(v)$. Is this dependence optimal? In particular, can one improve the $\varepsilon^{-1}$ factor, or obtain a bound in which the discrepancy-dependent term appears only additively in the runtime? More generally, does a matching lower bound hold beyond the limitations of augmenting path-based techniques?

\subparagraph{Open Problem II: Multi-way splitting.}
In our edge-coloring results we recursively apply \autoref{thm:balanced-local} in a divide-and-conquer fashion, yielding a $k$-way partition at a cost of $\bigO(\varepsilon^{-1}\log^2 k \log n)$ rounds. Is there a deterministic $\bigO(\varepsilon^{-1} \cdot \log n)$-round algorithm that directly partitions the edges into $k=\omega(1)$ parts such that each node $v$ has $\deg(v)/k \pm \varepsilon\Delta$ edges in each part? Achieving such a bound would remove the additional $\log^2 k$ overhead incurred by recursive splitting.

\subsection{Related Work}

\subparagraph{Degree splitting.}
The first authors to consider the problem of degree splitting in the distributed setting were Hanćkowiak, Karoński and Panconesi in 2001 \cite{HKP01}.
They used degree splits as a subroutine for their breakthrough result of designing the first polylogarithmic deterministic distributed maximal matching algorithm.
However, to achieve this result they had to work with a relaxed version of the problem, which allowed a $1/\poly \log n$ fraction of nodes to have arbitrary splits.
The first algorithm solving the degree splitting problem in full generality only came much later, when Ghaffari and Su released their influential paper \cite{GS_soda17} in 2017.
The paper is best-known for its seminal sinkless orientation algorithm,
but it also introduced the first distributed algorithms for directed and undirected degree splitting.
Surprisingly, their splitting algorithms do not use sinkless orientation as a subroutine---instead, they rely on an augmenting-path based approach. 
As a result, they obtain a deterministic directed degree splitting algorithm in $\bigO(\log^7 n / \varepsilon^3)$ rounds and a randomized counterpart in $\bigO(\log^4 n/\varepsilon^3)$ rounds.
For the undirected setting, they obtain a deterministic algorithm of complexity $\bigO(\log^7 n / \varepsilon^{3})$ for a sufficiently large $\Delta$ and a randomized algorithm of complexity $\bigO(\Delta^2 \cdot \log^4 n / \varepsilon^2)$, which works for all values of $\Delta$.
All of their algorithms compute a balanced split with discrepancy at most $\varepsilon \cdot \Delta$.
The next improvement, and the current state-of-the-art, comes from \cite{GHKMSU_dc20}, which leverages the sinkless orientation algorithm in \cite{GS_soda17} to improve the discrepancy of a given orientation.
Their deterministic algorithm relies on a recursive bootstrapping procedure and computes both directed and undirected splits with a discrepancy of $\varepsilon \cdot \deg(v)$ in  $\bigO(\varepsilon^{-1} \log \varepsilon^{-1} (\log \log \varepsilon^{-1})^{1.71} \cdot \log n )$ rounds.
Most recently, \cite{BKP_25} further generalized the above work by providing an algorithm for mixed directed and undirected problems with the same discrepancy. Moreover, they reduce the additive constant in the discrepancy for undirected splits from $4$ down to just $1$ for odd-degree vertices and $2$ for even-degree vertices.
Furthermore, in \cite{HMN22} a different kind of splitting is considered.
They partition the vertices of a graph into $k$ parts $V_1,\dots,V_k$ such that for each $i \in [k]$, each node $v \in V$ has at most $(\deg(v) + \varepsilon \Delta)/k$ neighbors in $V_i$.
They show that this problem can be solved in $\bigO(\varepsilon^{-1} + \poly \log \log n)$ rounds, while building on randomized algorithms for the distributed Lovász Local Lemma as a core subroutine. 

\subparagraph{(Hypergraph) sinkless orientation.}

The \emph{sinkless orientation} problem is arguably the simplest problem of complexity $\Theta(\log_\Delta n)$ in the \LOCAL model of distributed computing, on $\Delta$-regular graphs.
A lower bound of $\Omega(\log_\Delta \log n)$ on the randomized complexity was shown in~\cite{BFHKLRSU_stoc16} and an algorithm matching the corresponding lower bound was promptly given in \cite{GS_soda17}.
This bound on the randomized complexity was later extended to an $\Omega(\log_\Delta n)$ lower bound on the deterministic complexity of the problem~\cite{CKP_siamcomp19} matching the complexity of the deterministic algorithm in \cite{GS_soda17}.
The $\Omega(\log_\Delta n)$ lower bound for deterministic algorithms for sinkless orientation (and hence also for directed degree splitting) was later simplified \cite{BKKLOPPRSSSU_sinkless_made_simple_sosa23} and extended to the stronger supported LOCAL model, in which the communication graph $G$ is a subgraph of a graph $H$ known to all nodes.

Recently, it has been shown that the generalization of the sinkless orientation problem to hypergraphs, i.e.,  the \emph{hypergraph sinkless orientation (HSO)} problem, also has a $\Theta(\log n)$ intermediate complexity \cite{BMNSU_soda25}. In this work, we will use HSO as a major subroutine to design efficient degree splitting algorithms with small discrepancies. In the HSO problem, one is given a hypergraph, and the goal is to orient each hyperedge toward one of its incident vertices so that no vertex becomes a sink. The deterministic HSO algorithm from \cite{BMNSU_soda25} works under the condition that the minimum degree $\delta$ of the hypergraph exceeds its maximum rank $r$ by a constant factor. Under this condition, it runs in time $\bigO(\log_{\delta / r} n)$.

\subparagraph{Related work on edge coloring.}

Distributed edge coloring is one of the most extensively studied problems in the field of distributed computing. Hence we only present the results most relevant to our work here.
For a more thorough discussion of the current state-of-the-art we refer to the overview in the recent paper \cite{JMS25}.
For the regime we care about here, the authors of \cite{GKMU18} present the first deterministic $3\Delta / 2$-edge coloring algorithm of runtime $\bigO(\Delta^9 \cdot \poly \log n)$.
The runtime was subsequently improved multiple times.
The first improvement came via a speedup in one of the subroutines, when Harris \cite{Harris20} showed a faster algorithm for the hypergraph maximal matching problem.
This brought the runtime down to $\bigOtilde(\Delta^3 \log^4 n)$\footnote{We use the notation $\bigOtilde(f(n))$ to hide factors polylogarithmic in $f(n)$.}.
The current state of the art comes from \cite{BMNSU_soda25}, which uses
hypergraph sinkless orientation as a subroutine to get the runtime down to $\bigO(\Delta^2 \cdot \log n)$.
Importantly, they use a degree splitting technique from \cite{GHKMSU_dc20} to reduce the dependency on $\Delta$, which leads to a $\bigOtilde(\varepsilon^{-2} \cdot \log^2 \Delta \cdot \log n)$-round algorithm for $(3/2 + \varepsilon)\Delta$-edge coloring.

\subsection{Our Algorithm in a Nutshell}

For simplicity's sake, we assume $G = (V,E)$ to be a $\Delta$-regular graph and focus on the directed degree splitting problem.

\subparagraph{Split and orient.}
For each vertex $v \in V$, we create $\lfloor \Delta / 2 \rfloor $ virtual copies of degree $2$ and possibly one degree-$1$ node if $\Delta$ is odd. 
The resulting \textit{split graph} is a union of paths and cycles.
From a centralized perspective we can simply orient all paths and cycles consistently. 
The discrepancy at a vertex $v \in V$ is bounded by the number of paths originating at a copy of $v$.
Hence, the discrepancy will be $0$ for nodes of even degree and $1$ for nodes of odd degree. 

\subparagraph{From centralized to distributed.}
The main bottleneck with this approach is that consistently orienting paths---for the sake of this overview assume that the split graph has no cycles---in a distributed manner has an overhead that is proportional to the length of the path. The natural solution is to chop each long path into shorter ones. 

\subparagraph{Why is standard symmetry breaking not enough?}

A commonly used tool to break long paths into short ones in distributed  models is ruling sets \cite{CP19, Chang20, Chang24}.
Ruling sets are a relaxation of maximal independent sets
and come with a parameter $k$, which prescribes the maximal distance that any vertex can have to the ruling set.
Hence, by chopping the path at every ruling set vertex, we can obtain a sequence of shorter paths of maximum length $k$.
However, these chops come at a price.
At every \emph{chopping point} we introduce new start/endpoints of shorter paths and as a result there is no coordination between the orientation of these paths. As a result, the discrepancy of the vertex of $G$ hosting the split may increase by an additive two for every chopping point that it hosts. 
Since the classic symmetry breaking tools do not come equipped with any kind of load balancing features, it may happen that we are left with a very imbalanced split for some unlucky vertices. 

\subparagraph{Balancing chopping points via HSO.} 
Ensuring that we do not chop at too many virtual copies of the same vertex can be seen as a resource allocation problem.
Each virtual node on a long path or cycle wants a chopping point somewhere in their vicinity. Therefore we let each virtual node $v$ pick a subpath $P(v)$ of length $\ell$ that contains $v$. 
Next, we compute a ruling set on the intersection graph of these subpaths (where overlapping segments are adjacent) to produce a collection of disjoint paths $\mathcal{P}$.
Each path $P \in \mathcal{P}$ can now be viewed as a player competing for limited chopping points. 
The chopping points are limited because each vertex $v$ in the original graph $G$ divides its virtual copies into $x$ buckets of approximately equal sizes and allows only one node in each bucket to be chopped. 
This guarantees that the discrepancy does not exceed $2x$ for even-degree vertices and $2 x + 1$ for odd-degree vertices. 

Our goal is now to assign buckets to paths in $\mathcal{P}$ such that no bucket is assigned more than once and each path $P \in \mathcal{P}$ gets at least one bucket that contains a virtual node in $P$.
This problem can be modeled as a hypergraph sinkless orientation instance, where each path $P \in \mathcal{P}$ represents a vertex.
Each bucket then models a hyperedge which contains all the paths that include at least one node from the bucket.
The rank of a hyperedge can be as large as the number of nodes in its corresponding bucket and is thus tied to the discrepancy we want to achieve. If we set $x := \left\lfloor \frac{\varepsilon \Delta}{2} \right\rfloor $, then the maximum rank $r$ satisfies $r \in \Theta(\varepsilon^{-1})$.
Hence, in order to satisfy the necessary degree-rank condition that the minimum degree $\delta$ of the hypergraph $H$ must exceed the maximum rank $r$ our only option is to increase the minimum degree.
The length parameter $\ell$ can be tuned to increase the degree of the HSO instance and thereby satisfy the required degree-rank condition, but it also appears as a multiplicative overhead in the simulation cost of the hypergraph.
Therefore, we choose $\ell = \Theta(\varepsilon^{-1})$.
Since hypergraph sinkless orientation can be solved in $\mathcal{O}(\log n)$ time deterministically, this leads to a total runtime of $\mathcal{O}(\varepsilon^{-1} \cdot \log n)$.

\subparagraph{Towards optimal discrepancy.} The outlined approach results in a discrepancy of $2x+1$. Since we cannot set $x$ to be zero, it seems like the best discrepancy we can achieve is three.
Thus, for very balanced splits, an additional step is needed.
For each vertex of discrepancy three, we take its three virtual copies of degree one and merge them into another virtual node of degree three. A sinkless orientation on this graph finally reduces the discrepancy to one.

\subparagraph{Extension to undirected splitting.}
For the undirected splitting problem we do not aim to orient paths consistently but to color them red/blue in an alternating manner. Here, additional care has to be taken with short cycles, e.g., a cycle of length three does not allow for a red/blue coloring that is balanced for all of its vertices. Hence, we preprocess the split graph by locally rewiring the vertex splits such that all short cycles are removed. 
This procedure takes $\mathcal{O}(\ell)$ rounds and does not dominate the overall runtime.
Long cycles are not a problem because we will chop them in the middle anyway and thus we can deal with them just like long paths.

\section{Preliminaries}

For a hypergraph $H = (V,E)$, we consider its bipartite representation $B_H = (V \sqcup E, F)$, where the nodes and edges of $H$ form opposite sides of the bipartition and there is an edge $(v,e) \in F$ for every pair $(v,e) \in V \times E$ such that $v \in e$.
We use the classic convention of defining \LOCAL (or \CONGEST) on a hypergraph $H$ as \LOCAL (or \CONGEST) on $B_H$.

\subsection{Graph Notation}

We denote by $G=(V,E)$ the input graph, which also acts as communication network. 
Throughout this paper, we assume $n$ to be an upper bound on the number of nodes $\card{V}$, known to all nodes in the graph. $\Delta$ is an upper bound on the maximum degree, also assumed to be globally known. 
Every node $v \in V$ has a unique identifier $\ID_v$ between $1$ and $n$. For a node $v$, we denote by $N(v)$ its neighborhood in the graph ($N(v) = \set{u : uv \in E}$), and by $\incN(v)$ its inclusive neighborhood ($\incN(v) = N(v) \cup \set{v}$). 
Similarly, we write $E(v) = \{ e: v \in e \} $ to denote the set of all edges incident to $v$.
We write $\dist(u,v)$ to denote the distance between two nodes $u$ and $v$, defined as the length of the shortest path in $G$ between $u$ and $v$. 
When $u$ and $v$ are disconnected, we define their distance to be infinite, $\dist(u,v) = \infty$. 
We define the inclusive distance-$k$ neighborhood of a node $v$ as $\incN^k_G(v) = \set{u:\dist(u,v) \leq k}$, and the distance-$k$ neighborhood of $v$ as $N^k_G(v) = \incN^k(v)\setminus \set{v}$.
Further, we define $B^k_G(v)$ as the vertex-induced subgraph of $\incN^k_G(v)$. When it is clear from context we sometimes also omit $G$ from the subscript.
For any integer $n \in \mathbb{N}$, we write $[n]$ to denote the set $\{1,2,\dots,n\}$.

\subsection{Problem Definitions and Subroutines}

\begin{definition}[Directed Degree Splitting]
    Given a graph with maximum degree $\Delta$, orient each edge such that for each node, the discrepancy between in-degree and out-degree is at most $\varepsilon \cdot \Delta$.
\end{definition}

\begin{definition}[Undirected Degree Splitting]
    Given a graph with maximum degree $\Delta$, color each edge red or blue such that for each node, the discrepancy between incident red and blue edges is at most $\varepsilon \cdot \Delta$.
\end{definition}

We note that our algorithms in fact solve a slightly stronger variant of the splitting problems, where the discrepancy at each node of degree $\mathrm{deg}(v)$ is at most $\varepsilon \cdot \mathrm{deg}(v)$.

\begin{definition}[Sinkless Orientation]
    In the \emph{Sinkless Orientation} problem, the goal is to orient all edges such that every node has at least one outgoing edge. Put differently, the graph is free of sinks.
\end{definition}

Sinkless orientation stated in this form can actually be potentially unsolvable, or solvable with a very high complexity, on some graphs (notably paths and cycles). As such, the following problem is often considered instead.

\begin{definition}[min-degree-$\delta$ Sinkless Orientation]
    In the \emph{min-degree-$\delta$ Sinkless Orientation} problem, the goal is to orient all edges such that every node of degree $\delta$ or higher is not a sink.
\end{definition}

\begin{lemma}[Sink- and Sourceless Orientation {\cite[Corollary 3.2]{GHKMSU_dc20}}]
    \label{lem:sink-sourceless}
    A min-degree-$\delta$ sink- and source-less orientation for $\delta \geq 3$ can be computed in $\bigO(\log_\delta n)$ rounds of deterministic \LOCAL and with high probability in $\bigO(\log_\delta \log n)$ rounds of randomized \LOCAL.
\end{lemma}

In particular, Sinkless Orientation has complexity $\Theta(\log_\Delta n)$ and $\Theta(\log_\Delta \log n)$ on the graph class of $\Delta$-regular graphs, for $\Delta \geq 3$.

\begin{definition}[Hypergraph sinkless orientation (HSO) \cite{BMNSU_soda25}]
    \label{def:hso}
    The objective of the \emph{hypergraph sinkless orientation problem (HSO)} is to assign each hyperedge in a hypergraph $H$ to exactly one of its incident vertices such that every vertex is assigned at least one incident hyperedge.
\end{definition}

For each hyperedge $e$ in $H$, we define its rank $r(e)$ as the number of vertices in $e$.

\begin{theorem}[{deterministic HSO \cite[Theorem 1.4]{BMNSU_soda25}}]
    \label{thm:hso}
    There is a deterministic \LOCAL algorithm of complexity $\bigO(\log_{\delta / r} n)$ for computing a hypergraph sinkless orientation on any $n$-vertex multihypergraph $H$ with maximum rank $r$ and minimum degree $\delta > r$.
\end{theorem}

\begin{theorem}[{randomized HSO \cite[Theorem 1.5]{BMNSU_soda25}}]
    \label{thm:hso-rand}
    There is a randomized \LOCAL algorithm of complexity $\bigO(\log_{\delta / r} \delta + \log_{\delta / r} \log n)$ for computing a hypergraph sinkless orientation on any $n$-vertex multihypergraph $H$ with maximum rank $r$ and minimum degree $\delta > 320 r \log r$ with high probability.
\end{theorem}

\begin{definition}[$(\alpha,\beta)$-ruling set]
    For two integers $\alpha, \beta \geq 1$, an \emph{$(\alpha,\beta)$-ruling set} is a set of nodes $S \subseteq V$ such that the distance between any two nodes in $S$ is at least $\alpha$ and any node $v \in V$ has a distance of at most $\beta$ to the closest node in $S$.
\end{definition}

Ruling sets generalize maximal independent sets, which are $(2,1)$-ruling sets. 
In this paper, we will use the following ruling set algorithm obtained by combining results from \cite{SEW_RulingSetColoring_tcs13} and \cite{GV_RulingSetBoundedGrowth_podc07}.

\begin{theorem}[{\cite[Theorem 3]{SEW_RulingSetColoring_tcs13} \& \cite[Theorem 8]{GV_RulingSetBoundedGrowth_podc07}}]
    \label{thm:det_ruling_set}
    There is a deterministic distributed algorithm that computes a $(2,\bigO(\log \Delta))$-ruling set in any graph $G$ of maximum degree $\Delta$ in $\bigO(\log \Delta + \log^{\ast} n)$ rounds of the \LOCAL model.
\end{theorem}

\section{Directed Degree Splitting}
\label{sec:balanced-local}

In this section we describe our improved \LOCAL algorithm for directed degree splits that is obtained through a connection with the hypergraph sinkless orientation (HSO) problem.

\SplittingThm*

Setting $\varepsilon < 1 / \Delta$ we obtain the following corollary.

\begin{corollary}
    There is a deterministic \LOCAL algorithm of complexity $\bigO(\Delta \cdot \log n)$ that computes an orientation of the edges in any graph $G$ on $n$ nodes such that the discrepancy at each vertex $v$ is at most $2$ if $\deg(v)$ is even and at most $1$ if $\deg(v)$ is odd.
\end{corollary}

This improves on a previous result of complexity $\bigO(\Delta \cdot \log \Delta \cdot (\log \log \Delta)^{1.71}\cdot \log n)$ for the same problem by \cite{GHKMSU_dc20}.

\subsection{High-Level Overview}

Before we describe our algorithm, it is useful to revisit a simple corollary about balanced orientations from sink- and sourceless orientations. The proof of this statement can be found in \cite{GHKMSU_dc20}, but it is simple enough to restate it here.

\begin{corollary}
    \label{cor:balanced-orientation}
    Let $s\geq 3$ be an integer. An orientation where each node $v$ of degree $\deg(v)$ has at least $\floor{\deg(v)/s}$ outgoing and incoming edges can be computed in $\bigO(\log_{s} n)$ rounds of deterministic \LOCAL and with high probability in $\bigO(\log_{s} \log n)$ rounds of randomized \LOCAL.
\end{corollary}

\begin{proof}
    For each vertex $v$ of degree $\deg(v)$ we create $\max(1,\floor{ \deg(v) / s })$ copies and divide the edges incident to $v$ evenly among the copies. 
    Nodes of degree $s$ or more create copies of degree $s$ or more.
    By applying \autoref{lem:sink-sourceless} we obtain an orientation
    of this virtual graph where nodes of degree $s$ or more are
    neither sinks nor sources. Computing this orientation takes
    $\bigO(\log_s n)$ rounds deterministically and $\bigO(\log_s \log
    n)$ rounds randomized.  Lifting this orientation back to $G$, we
    get that each vertex $v$ has at least $\lfloor \deg(v) / s \rfloor
    $ incoming and outgoing edges.
\end{proof}

Since sinkless orientation only guarantees an outgoing edge for nodes of degree at least three, the discrepancy achievable with this approach is capped at $\deg(v) / 3$. 
Hence, in order to go below this barrier, we need to adopt a more aggressive splitting strategy. 
We split each vertex into virtual copies of degree $2$ (with one additional degree-one copy for odd-degree vertices).
Our splitting strategy comes with the clear upside that it is now very easy to compute a near-perfect split from a centralized perspective. All nodes on the same path/cycle decide on a consistent orientation of the path/cycle. This procedure clearly produces an orientation where no node of degree two is a source or a sink. 
Since the edges of the split graph are in a one-to-one correspondence with the edges of $G$, we can immediately lift the orientation back to $G$.
All even-degree vertices are thus perfectly split, while for odd-degree vertices we get an unavoidable discrepancy of one.

However, there is also a very clear downside of this approach.
Without any additional modifications, it is a very inefficient distributed algorithm.
The maximum length of a path in the split graph might be linear in the number of edges in $G$ and even the weak diameter---where distances are measured in the original graph $G$---might still be linear in the number of nodes in $G$.
Hence, consistently orienting every path and cycle in such a split graph might take up to $\Omega(n)$ rounds in \LOCAL.

Therefore, we need to construct a split graph, where each path or cycle has a small weak diameter in $G$.
We start with an arbitrary split graph of the input graph $G$.

\begin{definition}[$G_\mathrm{split}$]
    \label{def:G_split}
    We define the split graph $G_\mathrm{split} = (V_\mathrm{split}, E)$ as follows:
    Each vertex $v$ obtains a list $N(v) = [w_1,w_2,\dots,w_{\deg(v)}]$ of all its neighbors sorted in ascending order of IDs.
    Then, it creates $\lceil \deg(v) / 2 \rceil $ virtual copies and assigns the edge $(v,w_i)$ to its $\lceil i / 2 \rceil $-th virtual copy.
    Symmetrically, the node $w_i$ assigns the edge $(v,w_i)$ to its $\lceil j / 2 \rceil $-th virtual copy, if $v$ is the neighbor with $j$-th lowest ID of $w_i$.
    We denote the set of all virtual copies of $v$ by $\mathrm{split}(v)$.
    Further, for any $w \in \mathrm{split}(v)$ we call $v$ the \emph{creator} of $w$.
\end{definition}

To avoid confusion we will reserve the term \emph{vertex} to refer to vertices of $G$ and \emph{node} to refer to nodes of $G_\mathrm{split}$.
Note that the construction of $G_\mathrm{split}$ does not require coordination between the vertices of $G$ and further, one round on $G_\mathrm{split}$ can be simulated in one round of $G$ in \LOCAL.
Furthermore, when we talk about distances in the split graph, we sometimes need to refer back to the communication network $G$.

\begin{definition}[Weak diameter]
    Let $v,w \in V, v' \in \mathrm{split}(v)$ and $w' \in \mathrm{split}(w)$.
    Then, with slight abuse of notation we write $\dist_G(v',w')$ to refer to the distance between $v$ and $w$ in $G$.
    Furthermore, for any subgraph $H \subseteq G_\mathrm{split}$ we define the \emph{weak diameter} $\mathrm{diam}_G(H)$ as the diameter of the subgraph of $G$ induced by the creators of nodes in $H$.
\end{definition}

Now we want to modify $G_\mathrm{split}$ in such a way that all long paths and cycles disappear, while chopping at most $x(v) := \max \left\{1, \left\lfloor \frac{\varepsilon \deg(v)}{2} \right\rfloor \right\}$ degree-2 copies of each vertex $v \in V$.
The first goal could be achieved by computing a maximal independent set $S$ on the power graph $G_\mathrm{split}^k$.
Chopping the paths at every node $v \in S$ would produce paths of length at most $k$, but might lead to high discrepancies at some unlucky vertices.
Hence, we let each vertex split all its virtual copies into $x(v)$ buckets of roughly equal size and allow at most one vertex to be chopped in each bucket.

\begin{definition}[buckets]
    \label{def:buckets}
    Let $v$ be a vertex of degree $\deg(v)$ in $G, x(v) \in \mathbb{N}_+$, $m := \lceil \mathrm{deg}(v) / 2 \rceil $ and $[w_1,w_2,\dots,w_m]$ be a list of its virtual copies in $G_\mathrm{split}$. 
    Then, we define 
    \[
        R_i(v) := \{ w_j: \frac{i - 1}{x(v)} \cdot m \leq j <  \frac{i}{x(v)} \cdot m\}
    \] 
    as the \emph{$i$-th bucket} of $v$ for $i \in [x(v)]$.
    Further, we write $R(v) := \{ R_i(v): i \in [x(v)] \} $ and $\mathcal{R} = \bigcup_{v \in V} R(v)$.
\end{definition}

Note that each bucket satisfies $\lvert R_i(v) \rvert \in \{ \lfloor \mathrm{deg}(v) / (2x(v)) \rfloor, \lceil  \mathrm{deg}(v) / (2x(v)) \rceil \} $.

\begin{definition}[Hypergraph for HSO]
    \label{def:hypergraph}
    Let $\mathcal{P}$ be a set of pairwise disjoint subpaths of $G_\mathrm{split}$. Then we define a hypergraph $H(\mathcal{P}, \mathcal{R})$ on the vertex set $\mathcal{P}$ as follows:
    For each bucket $R_i(v) \in \mathcal{R}$ we add a hyperedge
    $e(R_i(v)) := \{ P \in \mathcal{P}: P \cap R_i(v) \neq \emptyset\} $.
\end{definition}

Recall that solving HSO efficiently requires a hypergraph $H$, where the minimum degree $\delta$ exceeds the maximum rank $r$.
The maximum rank of $H(\mathcal{P},\mathcal{R})$ is clearly bounded by the maximum size of a bucket $R \in \mathcal{R}$. 
For each bucket $R_i(v) \in \mathcal{R}$ we have that $\lvert R_i(v) \rvert \leq \left\lceil \frac{\deg(v)}{2 x(v)} \right\rceil$.
In order to prove \autoref{thm:balanced-local}, we will choose $x(v) := \max \left\{1, \left\lfloor \frac{\varepsilon \deg(v)}{2} \right\rfloor \right\}$.
Hence, we need to construct a set $\mathcal{P}$ such that every path $P \in \mathcal{P}$ intersects at least $\Omega(\varepsilon^{-1})$ buckets.

For that purpose we first let each node $v$ in $G_\mathrm{split}$ select a subpath $B(v) \ni v$ of its own component in $G_\mathrm{split}$. The subpath $B(v)$ should contain nodes from $\ell$ distinct buckets, where $\ell$ is a parameter that we can tune to satisfy the degree-rank condition.
Since these subpaths overlap a lot, we next compute a ruling set $\mathcal{P}$ on the intersection graph $G_\mathcal{B}$ of $\mathcal{B} = \{ B(v): v \in V_\mathrm{split} \}$.
This ensures that all paths in $\mathcal{P}$ are pairwise disjoint, while the remaining graph $G_\mathrm{split} \setminus \bigcup_{\mathcal{P}} P$ decomposes into a set of paths and cycles of small diameter.
Then, we can apply the deterministic HSO algorithm from \cite{BMNSU_soda25} to assign one bucket exclusively to each path in $\mathcal{P}$.
For each path $P \in \mathcal{P}$, let $\varphi(P)$ denote the bucket assigned to $P$ by the HSO procedure.
By construction of $H$, the path $P$ now contains a virtual node of the bucket $\varphi(P)$ that it can split into two virtual nodes of degree one.
Denote this new virtual graph obtained this way by $G'$.
Now, all paths and cycles in $G'$ have a small weak diameter in $G$ and we can efficiently compute a consistent orientation of all of them.
Finally, we observe that the discrepancy of any vertex $v \in V$ is at most $2 \cdot x(v) + 1$.
Each vertex $v$ is split into $x(v)$ buckets.
For each bucket, at most one virtual node is split into two degree-one nodes and all other nodes get one incoming and one outgoing edge.
Since odd-degree vertices have an additional degree-one node, this yields a total discrepancy of at most $2 \cdot x(v) + 1$.

\subsection{The Algorithm}
\label{sec:simple_algorithm}

In this section we prove the directed statement of \autoref{thm:balanced-local}.

\begin{algorithm*}[ht!]
	\caption{Directed degree splitting}
    \label{alg:directed_splits}
	\begin{algorithmic}[1]
	\Statex \textbf{Input:} A graph $G = (V,E)$, a parameter $\varepsilon > 0$.
    \Statex \textbf{Output:} An orientation of $G$ with discrepancy at most $\varepsilon \cdot \deg(v) + 2$ for all $v \in V$.
    \Statex 
    \State Construct an arbitrary split graph $G_\mathrm{split}$.
    \Comment{\autoref{def:G_split}}
    \State Set the parameters $x(v) := \max \left\{1, \left\lfloor \frac{\varepsilon \deg(v)}{2} \right\rfloor \right\}$ and $\ell := \lceil 8 \cdot \varepsilon^{-1} \rceil $.
    \State Divide virtual nodes into buckets $R_1(v),\dots,R_{x(v)}(v)$ for each $v \in V$.
    \Comment{\autoref{def:buckets}}
    \State Each node $v \in V_\mathrm{split}$ computes a local voting block $B(v)$.
    \Comment{\autoref{lem:local_voting_block}}
    \State $\mathcal{B} \gets \{ B(v): v \in V_\mathrm{split}\}, \mathcal{R} \gets \{ R_i(v): i \in [x(v)], v \in V \}$.
    \State Compute a subset $\mathcal{P} \subseteq \mathcal{B}$ of node-disjoint voting blocks.
    \Comment{\autoref{lem:ruling_set_voting_blocks}}
    \State Set up the hypergraph $H = H(\mathcal{P},\mathcal{R})$ between voting blocks $\mathcal{P}$ and buckets $\mathcal{R}$.
    \Comment{\autoref{def:hypergraph}}
    \State Compute a hypergraph sinkless orientation $\varphi: \mathcal{P} \to \mathcal{R}$ on $H$.
    \Comment{\autoref{thm:hso}}
    \State Every voting block $P$ splits exactly one node in $\varphi(P)$ into two degree-one nodes.
    \State Orient each path and cycle in this new graph consistently.
	\end{algorithmic}
\end{algorithm*}

\begin{definition}
    We define the \emph{true length} $\ell^\star(C)$ of a connected component $C$ in $G_\mathrm{split}$ as the number of distinct creators of nodes in $C$.
\end{definition}

Note that the weak diameter of a component always lower bounds its true length.

\begin{definition}[voting block]
    A \emph{voting block} of size $\ell$ is a connected subgraph $B$ of $G_\mathrm{split}$ of true length exactly $\ell$ and weak diameter at most $\ell$.
\end{definition}

After constructing $G_\mathrm{split}$, each virtual node $v$ tries to collect a voting block of size $\ell$.
Crucially, if such a voting block does not exist, then the path or cycle containing $v$ has to be short.

\begin{lemma}
    \label{lem:local_voting_block}
    Let $v \in V_\mathrm{split}$ and $C$ be the path or cycle in $G_\mathrm{split}$ that contains $v$. If the true length of $C$ is at least $\ell$, then there is a voting block $B(v) \subseteq C$ of size $\ell$ that contains $v$.
\end{lemma}

\begin{proof}
    We iteratively grow a voting block containing $v$ as follows:
    Define $B_0(v) := \{ v \} $ and for $i > 0$, add both neighbors of $B_i(v)$ in $C$ to $B_{i+1}(v)$. 
    Let $k \in \mathbb{N}$ be the smallest integer such that the true length of $B_k(v)$ is at least $\ell$.
    Note that $k$ always exists, since the true length of $C$ is at least $\ell$.
    Now observe that the creators of nodes in $B_i(v)$ induce a connected subgraph in $G$. Since the diameter of a graph is always a lower bound on its number of nodes, we get that the weak diameter of $B_k(v)$ is at most $\ell$.
    Since we add at most two nodes to $B_i(v)$ in each iteration, the true length of $B_k(v)$ is either $\ell$ or $\ell + 1$.
    In the latter case, we simply remove one of the two nodes added in the $k$-th iteration to obtain a voting block of size $\ell$.
\end{proof}

\begin{observation}
    \label{obs:true_length}
    The number of nodes in a voting block $B$ of size $\ell$ is bounded by $\lvert B \rvert \leq \ell^2$.
\end{observation}

\begin{proof}
    Let $V(B)$ denote the set of creators of nodes in $B$.
    Note that every edge in $G_{\mathrm{split}}[B]$ corresponds to an edge in the vertex-induced subgraph $G[V(B)]$.
    Since the true length of $B$ is $\ell$, the graph $G[V(B)]$ has at most $\ell$ vertices and thus contains at most $\binom{\ell}{2} \leq \ell^2 - 1$ edges.
    Hence, $G_{\mathrm{split}}[B]$ is a connected graph with at most $\ell^2 - 1$ edges, which implies that it can contain at most $\ell^2$ nodes.
\end{proof}

In general, the local voting blocks computed by the virtual nodes will overlap a lot.
This is problematic, since we need to bound the number of times that a single bucket is requested by a voting block.
Hence, we want to select a sparse subset $\mathcal{P}$ of non-overlapping voting blocks.
On the other hand, we want the subset $\mathcal{P}$ to be well spread out, in the sense that every node in $G_\mathrm{split}$ is still close to a voting block in $\mathcal{P}$.
The right tool for this job is to compute a ruling set on the intersection graph $G_\mathcal{B}$ of the set of all local voting blocks $\mathcal{B} = \{ B(v): v \in V_\mathrm{split} \} $.

\begin{lemma}
    \label{lem:ruling_set_voting_blocks}
    There is a deterministic distributed algorithm that computes a set $\mathcal{P}$ of pairwise disjoint paths in $G_\mathrm{split}$ in $\bigO(\ell \cdot (\log \ell + \log^\ast n))$ rounds of \LOCAL such that
    \begin{itemize}
        \item each path $P \in \mathcal{P}$ is a voting block of size at least $\ell$ and weak diameter at most $\ell$ and
        \item every component of $G_{\mathrm{split}} \setminus \bigcup_{\mathcal{P}} P$ has weak diameter less than $\ell$.
    \end{itemize}
\end{lemma}

\begin{proof}
    As a first step, we use the ruling set algorithm of \autoref{thm:det_ruling_set} to compute a subset $\mathcal{P} \subseteq \mathcal{B}$ of pairwise disjoint voting blocks.
    Next, we observe that two voting blocks $B(v), B(w) \in \mathcal{B}$ can only intersect if $v$ and $w$ lie in the same connected component $C$ of $G_\mathrm{split}$. Due to \autoref{obs:true_length} we know that each voting block in $\mathcal{B}$ contains at most $\ell^2$ nodes. Therefore, if the distance between $v$ and $w$ is at least $2 \ell^2$, their respective voting blocks cannot intersect.
    This shows that the maximum degree $\Delta_\mathcal{B}$ of the intersection graph $G_\mathcal{B}$ is bounded by $4\ell^2$.
    Since $\mathcal{P}$ is a $\bigO(\log \ell)$-ruling set on $G_\mathcal{B}$, each component of $G_\mathrm{split} \setminus \bigcup_{\mathcal{P}} P$ has weak diameter at most $\bigO(\ell \cdot \log \ell)$.
    As long as there exists a component $C$ with weak diameter at least $\ell$, we can always add a voting block $B \subseteq C$ of size at least $\ell$ and weak diameter at most $\ell$ to $\mathcal{P}$ due to \autoref{lem:local_voting_block}.
    After this procedure terminates, the set $\mathcal{P}$ fulfils both criteria.
    The ruling set algorithm runs in $\bigO(\log \ell + \log^\ast n)$ rounds on $G_\mathcal{B}$, while simulating a round on $G_\mathcal{B}$ incurs a multiplicative $\bigO(\ell)$ overhead. The second part of the algorithm only operates on components of weak diameter $\mathcal{O}(\ell \cdot \log \ell)$ and thus takes at most $\bigO(\ell \cdot \log \ell)$ rounds on $G$.
\end{proof}

Each such voting block now competes against the other blocks to win the coveted prize of being the only block, which is allowed to split one vertex from a certain bucket.
Formally, we consider the hypergraph $H = H(\mathcal{P},\mathcal{R})$ described in \autoref{def:hypergraph}. 

\begin{lemma}[HSO Condition]
    \label{lem:hso_condition}
    The minimum degree $\delta$ and the maximum rank $r$ of $H$ satisfy $\delta \geq 2r$.
\end{lemma}

\begin{proof}
    Let $v \in V$ and $R_i(v) \in \mathcal{R}$ be a bucket of virtual nodes.
    The rank of $R_i(v)$ can be upper bounded by the number of virtual nodes in $R_i(v)$, which by construction (\autoref{def:buckets}) is at most
    \[
        \ceil*{ \frac{\deg(v)}{2x(v)} }
        \leq \ceil*{ \frac{\deg(v)}{2 \max \{1, \floor{ \varepsilon \cdot \deg(v) / 2 } \} } } \leq \ceil{ 2 \varepsilon^{-1} } \leq 4 \varepsilon^{-1}.
    \]
    On the other hand, any voting block $P \in \mathcal{P}$ contains nodes from $\ell = \lceil 8 \varepsilon^{-1} \rceil$ distinct buckets. Thus, $\delta \geq \ell = \lceil 8 \varepsilon^{-1} \rceil \geq 2r$.
\end{proof}

Now we are ready to apply the deterministic HSO algorithm from \cite{BMNSU_soda25} to prove the directed statement in \autoref{thm:balanced-local}.

\begin{proof}[Proof of \autoref{thm:balanced-local} (directed)]
    We run the deterministic HSO algorithm (\autoref{thm:hso}) on $H$. 
    According to \autoref{lem:hso_condition}, it holds that $\delta \geq 2r$ and therefore the algorithm runs in $\bigO(\log n)$ rounds of \LOCAL.
    Recall that the edges of $H$ represent buckets of virtual nodes in $\mathcal{R}$ and the vertices of $H$ represent voting blocks in $\mathcal{P}$.
    In a sinkless orientation of the hypergraph $H$, each voting block has at least one outgoing edge, whereas a hyperedge is outgoing for exactly one of its incident vertices.
    In other words, the algorithm produces an injective function $\varphi: \mathcal{P} \to \mathcal{R}$.

    For any voting block $P \in \mathcal{P}$, there exists a node $v_P \in P \cap \varphi(P)$. Let $G'$ be the graph obtained by splitting all nodes $v_P$ for $P \in \mathcal{P}$ into two degree-one copies. Now we claim that each component $C$ in $G'$ has weak diameter at most $4 \cdot \ell$ in $G$.
    Let $v \in C$ be an arbitrary node. 
    If $C$ does not contain any block $P$, then by \autoref{lem:local_voting_block} its weak diameter must be less than $\ell$.
    Otherwise, if there is a block $P \in \mathcal{P}$ that contains $v$, then $\dist_G(v,v_P) \leq \ell$.
    Otherwise, there must be a block $P' \in \mathcal{P}$ such that $\dist(v,P') \leq \ell$ and thus $\dist(v,v_{P'}) \leq 2\ell$.
    Since the distance from any node in $C$ to an endpoint of $C$ is at most $2\ell$ in $G$, the weak diameter of $C$ has to be bounded by $4\ell$.

    Now we orient all paths and cycles in $C$ consistently.
    Hence, every node of degree 2 in $G'$ has one incoming and outgoing edge.
    For each vertex $v$ and each bucket $R_i(v)$ there is at most one virtual node that gets split into two degree-1 nodes in $G'$. 
    Hence, the discrepancy at $v$ is at most $2x(v)$ for vertices of even degree and $2 x(v) + 1$ for odd degree. Since $x(v) = \max \left\{1, \left\lfloor \frac{\varepsilon \deg(v)}{2} \right\rfloor \right\}$, this leads to a discrepancy of at most $\varepsilon \cdot \deg(v) + 1$ for $\varepsilon \cdot \deg(v) \geq 1$.
    Since we need to set $x \geq 1$ for our algorithm to work, we can achieve a discrepancy of $2$ for even-degree vertices and $3$ for odd-degree vertices at best.
    
    To reduce the discrepancy for odd-degree vertices further, we construct another virtual graph.
    Each vertex of discrepancy $3$ merges its three degree-one copies in $G'$ into a single node of degree $3$.
    The remaining nodes in $G'$ remain unchanged.
    This graph now contains virtual nodes with degrees between one and three.
    We now want to modify our orientation such that no node of degree two or three is a sink or a source.
    We claim that this can be done in $\bigO(\ell \cdot \log n)$ via a reduction to sinkless orientation.
    We contract all paths of degree-two nodes into a single edge and call the resulting graph $G''$.
    Since every component of $G''$ has weak diameter of $\bigO(\ell)$, a communication round on $G''$ can be simulated in $\bigO(\ell)$ rounds of $G$.
    Since $G''$ has minimum degree 3 now, we can run \autoref{lem:sink-sourceless} on $G''$ to compute a sinkless orientation on $G''$ in $\bigO(\log n)$ rounds on $G''$.
    Finally, we can lift the orientation back to the original graph $G$ and obtain a discrepancy of at most $\varepsilon \cdot \deg(v) + 1$ for all vertices of odd degree.

    Constructing the hypergraph $H$ takes $\bigO(\ell \cdot (\log \ell + \log^\ast n))$ rounds on $G$ according to \autoref{lem:ruling_set_voting_blocks}.
    Solving HSO on $H$ takes $\bigO(\log n)$ rounds using \autoref{thm:hso}, while simulating a round on $H$ takes $\bigO(\ell)$ rounds.
    Consistently orienting all paths and cycles takes just $\bigO(\ell)$ rounds, since the weak diameter of any component in $G''$ is bounded by $4 \cdot \ell$.
    Finally, computing a sinkless orientation of $G''$ takes $\bigO(\log n)$ rounds, while a round on $G''$ can be simulated in $\bigO(\ell)$ rounds on $G$.
    Since $\ell = \lceil 8 \cdot \varepsilon^{-1} \rceil $, this yields a total runtime of $\bigO(\ell \cdot (\log \ell + \log n)) = \bigO(\varepsilon^{-1} \cdot \log n)$.
\end{proof}

\subsection{Randomized Algorithm}

In order to obtain a faster randomized version of our algorithm, we make three crucial changes:
Firstly, we replace the deterministic ruling set algorithm by a faster randomized algorithm.
Secondly, we replace the deterministic sinkless orientation algorithm by its randomized counterpart.
Finally, we replace the deterministic HSO algorithm by its faster randomized counterpart.
The randomized HSO algorithm in \cite{BMNSU_soda25} requires that $\delta \geq 320 r \log r$ (see \autoref{thm:hso-rand}). 
In order to optimize the runtime of our randomized algorithm, we extend the set of instances that are solvable in $\mathcal{O}(\log \log n)$ time by proving a tighter degree-rank condition of $\delta \geq 10^{\log^{\ast} r} \cdot r$.

We obtain this result through a recursive \emph{shattering} strategy. At a high level, we downsample a constant fraction of all hyperedges, that is, we replace the sets of nodes of some hyperedges with small random subsets of their original sets.
We do so in a manner that ensures that (1) for each node $v$, a constant fraction of the edges incident to $v$ are not downsampled, (2) the hypergraph induced by downsampled hyperedges has much lower rank than the original hypergraph, and (3) in this new hypergraph, most nodes have a degree that exceeds the rank as in the original graph, so the argument can be applied recursively to ultimately obtain an instance of Sinkless Orientation. The edges that are not downsampled are used to satisfy vertices which were not successful in the downsampling.

Each iteration of our shattering procedure is essentially a generalization of the shattering procedure used by the state-of-the-art randomized LOCAL algorithm by Ghaffari and Su for Sinkless Orientation~\cite{GS_soda17}.

\betterHSOthm*

\begin{proof}
    Consider a hypergraph $H = (V,E)$ with $\delta \geq 2 x \cdot 10^{\log^{\ast} r} \cdot r$. Without loss of generality we assume that $H$ is $\delta$-regular with $\delta = 2 x \cdot 10^{\log^{\ast} r} \cdot r$. This can be achieved by letting each vertex drop out of all but $2x \cdot 10^{\log^{\ast} r} \cdot r$ of its incident hyperedges. Henceforth, we will assume all hypergraphs appearing in this proof to be regular, since we may drop edges at any step of the recursion.
    Next, we note that if $\delta \geq 2 \log n \cdot r$, then the $0$-round algorithm where every hyperedge orients itself uniformly at random succeeds with high probability, as the probability that a vertex does not receive an outgoing hyperedge is at most $(1 - \frac{1}{r})^{\delta} \leq \exp(- \delta / r) \leq 1/n^2$.
    Thus, for the remainder of the proof we assume $\delta = \mathcal{O}(\log n \cdot r)$.
    As another preprocessing step, we let each hyperedge $e \in E$ sample a subset $e' \subseteq e$ of vertices of size $\Theta(\log n)$ to keep. 
    We show that this initial subsampling only loses a constant factor in the degree-rank ratio with high probability.

    \proofsubparagraph{Initial subsampling.}

    We let each hyperedge $e \in E$ of rank $r(e) > 16 \log n$ select a subset $e' \subseteq e$ of nodes of size exactly $16 \log n$ uniformly at random.
    For every edge $e$ of rank $r(e) \leq 16 \log n$ we set $e' := e$.
    Now we consider the hypergraph $H' = (V, E' := \{ e': e \in E \} )$.
    The maximum rank $r(H')$ is now bounded by $16 \log n$.
    Further, for any edge $e \in E$, the probability that a vertex $v \in e$ remains in the subsampled edge $e'$ is at least $\min\{1, 16 \log n / r(e) \} \geq 16 \log n/ r$.
    Therefore, the expected degree of any vertex $v$ is at least $16 \log n / r \cdot \delta$.
    Since the random choices of all hyperedges are mutually independent, we may 
    apply a standard Chernoff bound (\autoref{lem:chernoff}) to get that
    \[
        \mathbb{P}\left[\mathrm{deg}_{H'}(v) \leq \frac{1}{2} \cdot \mathbb{E}[\mathrm{deg}_{H'}(v)]\right] \leq \exp\left( - \frac{16 \log n }{8} \cdot \frac{\delta}{r}\right) \leq n^{-2}
    \]
    for all $v \in V$.
    Together with a union bound we get that the minimum degree $\delta(H')$ of $H'$ satisfies 
    \[
        \delta(H') \geq 8 \log n \cdot \frac{\delta}{r} \geq \frac{r(H')}{2} \cdot 2x \cdot 10^{\log^{\ast} r} \geq r(H') \cdot x \cdot 10^{\log^{\ast} r(H')}
    \]
    with high probability.

    Now we want to apply the subsampling process recursively.
    However, for sublogarithmic $r$ we do not get high probability guarantees on the concentration of vertex degrees anymore.
    Instead, we will show that the nodes where the process fails form small connected components, thus shattering the graph.
    We set $H_1 = H'$ and describe a random process that takes any hypergraph $H_i = (V_i, E_i)$ satisfying the degree-rank condition \eqref{eq:degree-rank-condition} and produces a subsampled instance $H_{i+1}^{\mathrm{sub}}$ of maximum rank $r(H_{i+1}^{\mathrm{sub}}) = \log r(H_{i})$ of successful vertices satisfying \eqref{eq:degree-rank-condition}, as well as a subgraph $H_{i+1}^{\mathrm{bad}}$ of failed vertices. 
    The first two steps of the random process described below ensure that the hypergraph $H_{i+1}^{\mathrm{bad}}$ of failed nodes can still be solved by the deterministic HSO algorithm of \cite{BMNSU_soda25}.

    \proofsubparagraph{Pre-shattering random process.}
    Let $H_i = (V_i,E_i)$ be a hypergraph satisfying 
    \begin{equation} \label{eq:degree-rank-condition}
        \delta(H_i) \geq x \cdot 10^{\log^{\ast} r(H_i)} \cdot r(H_i)
    \end{equation}
    and $\log^{\ast} r(H_i) \geq 5$.
    \begin{enumerate}
        \item Mark each hyperedge $e \in E_i$ independently with probability $3/4$.
        \item For any vertex with more than $7/8 \cdot \delta(H_i)$ or less than $1/2 \cdot \delta(H_i)$ marked incident edges, unmark all incident edges.
        \item For each remaining marked edge $e$, select a subset $e_{\mathrm{sub}} \subseteq e$ of nodes of size $\log r(H_i)$ uniformly at random.
    \end{enumerate}
    Now we define $E^\star$ as the set of unmarked edges, $E_{\mathrm{sub}} := \{ e_{\mathrm{sub}}:  e \in E_i \setminus E^\star \}$ and set
    \begin{itemize}
        \item $\mathrm{Bad}_1$ as the set of nodes where at least one of their incident hyperedges got unmarked in Step 2 and
        \item  $\mathrm{Bad}_2$ as the set of nodes not in $\mathrm{Bad}_1$ with less than $1 / 4 \cdot \delta(H_i) \cdot \log r(H_i) / r(H_i)$ incident edges in $E_\mathrm{sub}$.
    \end{itemize}
    Each node in $V_\mathrm{bad} := \mathrm{Bad}_1 \cup \mathrm{Bad}_2$ goes to the post-shattering instance $H_{i+1}^{\mathrm{bad}} = (V_\mathrm{bad}, E^\star|_{V_\mathrm{bad}})$, while each node in $V_\mathrm{sub} := V_i \setminus V_\mathrm{bad}$ goes to the subsampled instance $H_{i+1}^{\mathrm{sub}} = (V_\mathrm{sub},\{ e \cap V_\mathrm{sub}: e \in E_\mathrm{sub} \} )$.
    Next we show that this random process indeed shatters the graph and $H_{i+1}^{\mathrm{bad}}$ consists of components of size at most $\mathrm{poly}(\delta(H_i)) \cdot \log n$.

    \proofsubparagraph{Probability of bad events.}
    Let $\mathrm{marked}(v)$ denote the number of marked edges incident to $v$.
    Then, using a standard Chernoff bound (\autoref{lem:chernoff}) we get that
    \begin{align*}
        \mathbb{P}[\mathrm{marked}(v) \geq 7 / 8 \cdot \delta(H_i)] 
        &= \mathbb{P}[\mathrm{marked}(v) \geq 7/6 \cdot \mathbb{E}[\mathrm{marked}(v)]] \\
        &\leq \exp\left( - \frac{(1/6)^2 \cdot 3 / 4 \cdot \delta(H_i)}{2 + 1/6}\right) \leq \exp(- \delta(H_i) / 60).
    \end{align*}
    Further,
    \begin{align*}
        \mathbb{P}[\mathrm{marked}(v) \leq 1/2 \cdot \delta(H_i)] &=
        \mathbb{P}[\mathrm{marked}(v) \leq 2/3 \cdot \mathbb{E}[\mathrm{marked}(v)]] \\ 
        &\leq \exp\left( - \frac{(1 / 3)^2 \cdot 3 / 4 \cdot \delta(H_i)}{2}\right) = \exp(- \delta(H_i) / 24).
    \end{align*}
    Let $\mathrm{Unmarked}(v)$ denote the event that $\mathrm{marked}(v) \not\in [1/2 \cdot \delta(H_i), 7/8 \cdot \delta(H_i)]$. 
    We note that a vertex $v$ gets added to $\mathrm{Bad}_1$ either if $\mathrm{Unmarked}(v)$ occurs, or if for some $u \in N(v)$, $\mathrm{Unmarked}(u)$ occurs.
    Therefore, for $\delta(H_i) \geq 10000$ we have
    \begin{align*}
        \mathbb{P}[v \in \mathrm{Bad}_1] &\leq \mathbb{P}[\mathrm{Unmarked}(v)] + \sum_{u \in N(v)} \mathbb{P}[\mathrm{Unmarked}(u)] \\  
        &\leq 2 \exp(- \delta(H_i) / 60) + 2 \delta \cdot r(H_i) \exp(- \delta(H_i) / 60) \leq \frac{1}{\delta(H_i)^{16}}.
    \end{align*}
    Next, we note that any vertex $v \in \mathrm{Bad}_2$ is incident to at least $1/2 \cdot \delta(H_i)$ marked edges. Let $d(v)$ denote the number of edges in $E_\mathrm{sub}$ containing $v$.
    For each marked edge $e \in E_i$, the probability that $v$ is selected in the subset $e_\mathrm{sub} \subset e$ is at least $(\log r(H_i)) / r(H_i)$ and therefore $\mathbb{E}[d(v)] \geq \frac{\delta(H_i)}{2} \cdot \frac{\log r(H_i)}{r(H_i)}$. Another standard Chernoff bound then shows
    \begin{align*}
        \mathbb{P}[v \in \mathrm{Bad}_2] &\leq \mathbb{P}[d(v) \leq 1/4 \cdot \delta(H_i) \log r(H_i) / r(H_i)] \leq 
        \mathbb{P}[d(v) \leq 1/2 \cdot \mathbb{E}[d(v)]] \\ 
        &\leq \exp\left( - \frac{\delta(H_i) \log r(H_i)}{16r(H_i)} \right) \leq \exp\left( - \frac{x 10^{\log^{\ast} r(H_i)} \cdot \log r(H_i)}{16} \right) \\  
        &\leq \exp\left( - 1000 \frac{\log \delta(H_i)}{16} \right) \leq \frac{1}{\delta(H_i)^{16}}.
    \end{align*}

    \proofsubparagraph{$H_{i+1}^{\mathrm{bad}}$ shatters into small components}

    Next, we remark that each random variable $\mathbbm{1}[v \in \mathrm{Bad}_1]$ depends only on the randomness of its own incident hyperedges and the hyperedges incident to neighbors of $v$. Further, each random variable $\mathbbm{1}[v \in \mathrm{Bad}_2]$ depends only on the randomness of its own incident hyperedges.
    Hence, in the bipartite representation of the hypergraph, each bad event for $v \in V$ depends only on the randomness of nodes within at most $3$ hops from $v$.
    The maximum degree of this bipartite graph is $\delta(H_i)$.
    Therefore, we may apply \autoref{lem:shattering} with constants $c_1 = 16$ and $c_2 = 3$, which yields that the size of any connected component in $H_{i+1}^{\mathrm{bad}}$ is at most $\mathcal{O}(\mathrm{poly}(\delta(H_i)) \cdot \log n)$. 
    We will later show that we can solve HSO on each component in parallel and combine the local solutions into a globally feasible solution.

    \proofsubparagraph{Recursing on the subsampled instance $H_{i+1}^{\mathrm{sub}}$.}
    Observe that by definition of $\mathrm{Bad}_2$ every vertex in $H_{i+1}^{\mathrm{sub}}$ has degree at least
    \[
        \delta(H_{i+1}^{\mathrm{sub}}) \geq
        \frac{\delta(H_{i}^{\mathrm{sub}}) \cdot \log r(H_{i}^{\mathrm{sub}})}{4 r(H_{i}^{\mathrm{sub}})} \geq 
        x \cdot \frac{10^{\log^{\ast} r(H_{i}^{\mathrm{sub}})}}{4} \cdot \log r(H_{i}^{\mathrm{sub}}) \geq
        x \cdot 10^{\log^{\ast} r(H_{i+1}^{\mathrm{sub}})} \cdot r(H_{i+1}^{\mathrm{sub}}),
    \]
    Therefore, the degree-rank condition is satisfied for $H_{i+1}^{\mathrm{sub}}$ and we may recurse the rank-reduction procedure on $H_{i+1} := H_{i+1}^{\mathrm{sub}}$ until the rank satisfies $\log^{\ast} r(H_{i+1}) = 4$.

    \proofsubparagraph{Solving the base case.}
    Let $k$ be the integer such that $\log^{\ast} r(H_k) = 4$.
    Since $\log^{\ast} r(H_k) = 4$ we have that $\log r(H_k) \leq 16$ and thus
    \[
        \delta(H_k) \geq x \cdot 10^{\log^{\ast} r(H_k)} \cdot r(H_k) \geq 10000 \cdot r(H_k) \geq 320 r(H_k) \log r(H_k).
    \]
    Thus, the degree-rank condition to apply the already known randomized HSO algorithm (\autoref{thm:hso-rand}) is satisfied. Since $\delta$ and $r$ are constants now, we can compute an HSO on $H_k$ in $\bigO(\log \delta(H_k) + \log_{\delta(H_k) / r(H_k)} \log n) = \bigO(\log_x \log n)$ rounds.

    \proofsubparagraph{Solving the post-shattering instances $H_i^{\mathrm{bad}}$.}
    Note that Step 2 of the pre-shattering random process guarantees that no vertex in $H_i^{\mathrm{bad}}$ can have more than $7 / 8 \cdot \delta(H_{i-1})$ incident marked edges.
    Thus, the minimum degree of $H_{i}^{\mathrm{bad}}$ is at least $1 / 8 \cdot \delta(H_{i-1})$, while the maximum rank is still at most $r(H_{i-1})$.
    Since $H_{i-1}$ satisfies the degree-rank condition \eqref{eq:degree-rank-condition} and $r(H_{i-1}) > 2$ we get that
    \[
        \delta(H_i^{\mathrm{bad}}) \geq \frac{\delta(H_{i-1})}{8} \geq x \cdot \frac{10^{\log^{\ast} r(H_{i-1})}}{8} r(H_{i-1}) \geq x \cdot r(H_{i-1}) \geq x \cdot r(H_i^{\mathrm{bad}}).
    \]
    Furthermore, due to \autoref{lem:shattering} we get that the size of each component in $H_{i}^{\mathrm{bad}}$ is bounded by $\poly(\delta) \cdot \log n$. 
    Hence, we may apply the deterministic HSO algorithm (\autoref{thm:hso}) on each component of $H_{i}^{\mathrm{bad}}$ and for each $i = 2,\dots,k$ in parallel.
    Due to our initial subsampling we have that $r(H_{i}^{\mathrm{bad}}) = \mathcal{O}(\log n)$. In addition, we have that $\delta(H_{i}^{\mathrm{bad}}) = \mathcal{O}(r \cdot \log n)$.
    Therefore the runtime can be bounded by $\bigO(\log_x \delta(H_{i}^{\mathrm{bad}}) + \log_x \log n) = \mathcal{O}(\log_x \log n)$, where the $x$ in the base of the logarithm comes from the fact that $\delta(H_{i}^{\mathrm{bad}})/ r(H_{i}^{\mathrm{bad}}) \geq x$.

    \proofsubparagraph{Correctness.}
    First, we note that each vertex $v \in V(H)$ either moves to a post-shattering instance $H_i^{\mathrm{bad}}$ during the iterative rank-reduction procedure, or it remains in the final base case $H_k$. 
    In either case, $v$ receives at least one outgoing hyperedge by the HSO algorithms we run on $H_i^{\mathrm{bad}}$ and $H_k$.
    Further, observe that if an edge $e \in E(H)$ gets unmarked at any point in the procedure its subsampled version moves to the corresponding post-shattering instance and does not participate in the rank reduction anymore.
    Hence, any edge is marked outgoing for exactly one of its incident vertices.
    Thus, the orientation is a valid hypergraph sinkless orientation on $H$.

    \proofsubparagraph{Overall runtime.}
    Executing the pre-shattering random process and determining which vertices move on to the subsampled instance only takes a constant number of rounds. The number of subsampling iterations is bounded by $\mathcal{O}(\log^{\ast} r)$, since each iteration decreases $\log^{\ast} r(H_i)$ by one.
    Finally, solving the subsampled base case takes $\mathcal{O}(\log_x \log n)$ rounds randomized.
    Also, all the post-shattering instances can be solved in parallel by a deterministic algorithm in $\mathcal{O}(\log_x \log n)$ rounds.
    In the edge case that $\mathcal{O}(\log^{\ast} r)$ is not dominated by $\mathcal{O}(\log_x \log n)$, we stop the recursion earlier: 
    Let $k$ be the smallest integer such that either $\log^{\ast} r(H_k) \leq 4$ or $\log(r(H_k)) \leq x$ and observe that $k = \mathcal{O}(\log^{\ast} r - \log^{\ast} x)$. Then,
    \[
        \delta(H_k) \geq x \cdot 10^{\log^{\ast} r(H_k)} \cdot r(H_k) \geq 320 \log(r(H_k)) \cdot r(H_k)
    \]
    and we can compute an HSO on $H_k$ in $\mathcal{O}(\log_x \log n)$ rounds.
    Finally, we note that if $x \geq \log^{(k)} n$ for any constant $k$, then $\log^{\ast} r - \log^{\ast} x \in \mathcal{O}(1)$.
    Thus, the runtime is in fact dominated by $\mathcal{O}(\log_x \log n)$.
\end{proof}

Next, we develop a faster randomized version of \autoref{lem:ruling_set_voting_blocks}.

\begin{lemma}
    \label{lem:ruling_set_voting_blocks_randomized}
    There is a randomized distributed algorithm with round complexity {$\bigO(\ell \cdot \log \log n)$} that computes a set $\mathcal{P}$ of pairwise disjoint paths in $G_\mathrm{split}$ such that
    \begin{itemize}
        \item each path $P \in \mathcal{P}$ is a voting block of size at least $\ell$ and weak diameter at most $\ell$ and
        \item every component of $G_\mathrm{split} \setminus \bigcup_{\mathcal{P}} P$ has weak diameter less than $\ell$.
    \end{itemize}    
\end{lemma}

\begin{proof}
    As a first step, we use the ruling set algorithm of \autoref{cor:rand_ruling_set} to compute a subset $\mathcal{P} \subseteq \mathcal{B}$ of pairwise disjoint voting blocks.
    Due to \autoref{obs:true_length}, each voting block $B \in \mathcal{B}$ intersects with at most $3 \ell^2$ other voting blocks in $\mathcal{B}$.
    Hence, we can bound the maximum degree $\Delta_\mathcal{B}$ of $G_\mathcal{B}$ by $\Delta_\mathcal{B} = \bigO(\ell^2)$.
    Since $\mathcal{P}$ is a $\bigO(\log \log n)$-ruling set on $G_\mathcal{B}$, each component of $G_\mathrm{split} \setminus \bigcup_{\mathcal{P}} P$ has weak diameter at most $\bigO(\ell \cdot \log \log n)$.
    As long as there exists a component $C$ with weak diameter at least $\ell$, we can always add a voting block $B \subseteq C$ of size at least $\ell$ and weak diameter at most $\ell$ to $\mathcal{P}$ due to \autoref{lem:local_voting_block}.
    After this procedure terminates, the set $\mathcal{P}$ fulfills both criteria.
    The runtime of this procedure is at most the weak diameter of any component in $G_\mathrm{split} \setminus \bigcup_{\mathcal{P}} P$, which we bounded by  $\bigO(\ell \cdot \log \log n)$.
    The ruling set algorithm of \autoref{cor:rand_ruling_set} runs in $\bigO(\log \log n)$ rounds on $G_\mathcal{B}$, while simulating a round on $G_\mathcal{B}$ takes $\bigO(\ell)$ rounds.
\end{proof}

In order to apply \autoref{thm:better-rand-HSO}, we just need to increase the size of our voting blocks a bit compared to the deterministic setting. 
Indeed, we now build an HSO instance s.t.\ $\delta \geq 2 \cdot 10^{\log^{\ast} r} \cdot r$, when before we only had that $\delta \geq 2 r$.
Since the maximum rank of our hypergraph is tied to the discrepancy we want to achieve, our only option to still satisfy the more demanding degree-rank condition of \autoref{thm:better-rand-HSO} is to increase the size of our voting blocks. Since $r = \Theta(\varepsilon^{-1})$, this leads to a choice of $\ell = \Theta(\varepsilon^{-1} \cdot 10^{\log^{\ast} \varepsilon^{-1}})$.

\SplittingThmRand*

\begin{proof}[Proof of \autoref{thm:balanced-local-rand} (directed)]
    Let $\mathcal{P}$ denote the set of pairwise disjoint paths in $G_\mathrm{split}$ computed in \autoref{lem:ruling_set_voting_blocks_randomized} and $\mathcal{R}$ be the set of buckets defined in \autoref{def:buckets}. Let $H$ denote the hypergraph $H(\mathcal{P}, \mathcal{R})$ constructed in \autoref{def:hypergraph}. 
    In contrast to the deterministic version of our algorithm we now choose $\ell = \lceil 8 \cdot 10^{\log^{\ast} \varepsilon^{-1} } \cdot \varepsilon^{-1} \rceil \geq 2 \cdot 10^{\log^{\ast} r} \cdot r$. 
    The construction of $H$ takes $\bigO(\ell \cdot \log \log n)$ rounds.
    Now we apply \autoref{thm:better-rand-HSO} to solve HSO in $\mathcal{O}(\log \log n)$
    rounds. Recall that simulating a round on $H$ takes $\bigO(\ell)$ rounds on $G$ and thus the total runtime to compute an HSO on $H$ is also $\bigO(\ell \cdot \log \log n)$.
    Finally, by using the randomized sinkless orientation algorithm from \cite{GS_soda17}, we can reduce the discrepancy of odd-degree vertices to $1$ in just $\bigO(\ell \cdot \log \log n)$ rounds as well.
    The final step of computing a sinkless orientation on the virtual graph $G''$---whose construction we detail at the end of the proof of \autoref{thm:balanced-local} ---now takes just $\bigO(\ell \cdot \log \log n)$ rounds.
    Thus, the total round complexity of our algorithm is $\bigO(\ell \cdot \log \log n) = \bigO(\varepsilon^{-1} \cdot 10^{\log^{\ast} \varepsilon^{-1}} \cdot \log \log n)$.
\end{proof}

\section{Undirected Degree Splitting}
\label{sec:undirected-splits}

In this section, we prove the undirected parts of \autoref{thm:balanced-local} and \autoref{thm:balanced-local-rand}.
Our algorithm for directed splitting almost works out-of-the-box for undirected splitting as well.
There is only one small issue: Odd cycles cannot be properly $2$-colored.
Hence, any odd cycle that is left at the end of our algorithm will introduce an additional unit of discrepancy somewhere.
To get around this issue, we prepend a preprocessing step which removes all short cycles from the virtual graph $G_\mathrm{split}$ described in \autoref{def:G_split}.

\subsection{Removing Short Cycles}

We remove short cycles by merging them with adjacent cycles and paths.
Merging two edge-disjoint adjacent cycles results in a larger cycle, while merging a cycle with an adjacent path results in a larger path.

\begin{definition}[Short cycle]
    Let $C$ be a cycle in $G_\mathrm{split}$ and $X$ be a cycle or a path in $G_\mathrm{split}$. We define the distance between $C$ and $X$ as $\dist(C,X) := \min_{v \in C, v' \in X} \dist_G(v,v')$.
    We call a cycle $C$ in $G_\mathrm{split}$ short if the true length $\ell^\star(C)$ of $C$ is less than $\ell$ and long otherwise. 
\end{definition}

\begin{lemma}
\label{lem:cycle_doubling}
    Let $v \in V(G)$ and let $C$ denote a short cycle in $G_\mathrm{split}$ that contains at least one virtual copy $v_C$ of $v$. Then, for any virtual copy $v_X$ of $v$ that belongs to a different connected component $X$ in $G_\mathrm{split}$ we can locally change the split at $v$ such that $v_C$ and $v_X$ are connected in $G_\mathrm{split}'$ in a component of size $\lvert C \rvert + \lvert X \rvert $.
    If $X$ is a cycle, the new component is also a cycle. If $X$ is a path, the new component is also a path.
\end{lemma}

\begin{proof}
    Let $r_C (r_X)$ and $\ell_C (\ell_X)$ denote the two neighbors of $v_C (v_X)$ in $G_\mathrm{split}$. Since the edge $(v_C, r_C)$ closes a cycle, we can exchange it for the edge $(v_C, r_X)$ without disconnecting any vertices. At the same time we exchange the edge $(v_X, r_X)$ by the edge $(v_X, r_C)$. This connects the previously disconnected components $C$ and $X$, see \autoref{fig:cycle_joining} for an illustration.
\end{proof}

\begin{figure}
    \begin{subfigure}[t]{0.45 \textwidth}
        \centering
        \includegraphics{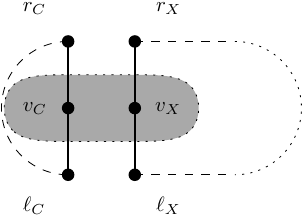}
        \caption{Two disjoint short cycles incident on two distinct virtual copies of the same vertex in $G$.}
    \end{subfigure}
    \hfill
    \begin{subfigure}[t]{0.45 \textwidth}
        \centering
        \includegraphics{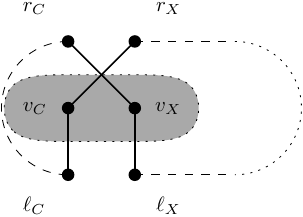}
        \caption{By interchanging two edges between the virtual nodes the two short cycles are merged into one larger cycle.}
    \end{subfigure}
    \caption{Doubling the length of a single short cycle in $G_\mathrm{split}$.}
    \label{fig:cycle_joining}
\end{figure}

The basic idea is now to execute \autoref{lem:cycle_doubling} for all vertices of $G$ in parallel. However, doing this naively runs into issues. Imagine we have three cycles $C_1,C_2,C_3$, where $C_1$ wants to merge with $C_2$, $C_2$ wants to merge with $C_3$ and $C_3$ wants to merge with $C_1$.
After executing the merges $C_1 \cup C_2$ and $C_2 \cup C_3$, we are left with just a single large cycle. If we now try to merge $C_3 \cup C_1$, then we would end up splitting the large cycle into two smaller cycles again.
This scenario needs to be avoided at all costs.
In other words, if we consider a virtual graph $G_\mathrm{join}$, where each short cycle is a node and there is an edge between two cycles that we want to join, then we must avoid cycles in $G_\mathrm{join}$.

\begin{lemma}
    \label{lem:increase_girth}
    Let $\ell \leq D$ be a positive integer, with $D$ the diameter of the graph.
    Then, there is a deterministic distributed algorithm of round complexity $\bigO(\ell)$ that computes a split of each vertex of degree $\deg(v)$ into $\lfloor \deg(v) / 2 \rfloor $ virtual copies of degree two (and one copy of degree one for vertices of odd degree) such that each cycle in the split graph has true length at least $\ell$.
\end{lemma}

\begin{proof}
    We define an evaluation function $f$ on every component of $G_{\mathrm{split}}$ as follows:
    For any short cycle $S$ we set $f(S) := \ell^\star(S)$, while for any long cycle $C$ we set $f(C) := \infty$. We also set $f(P) := \infty$ for any path $P$ in $G_\mathrm{split}$. 
    Each vertex $v \in V$ starts by determining the set of short cycles $\mathcal{S}(v)$ containing at least one virtual copy of $v$.
    This can be done in $\bigO(\ell)$ rounds of \LOCAL.
    For each short cycle $S \in \mathcal{S} = \bigcup_{v \in V} \mathcal{S}(v)$ we write $\mathrm{Min}(S)$ to denote the vertex with minimum ID among all creators of nodes in $S$ and let $\mathrm{Min}(S)$ handle all computations regarding $S$.
    Now each cycle $S$ finds a nearby path or the largest nearby cycle $\mathrm{Max}(S) := \mathrm{argmax} \{ f(X): \dist(X,S) \leq 2 \ell \} $. To break ties we use the IDs of $\mathrm{Min}(S)$ concatenated with the smaller ID of the two creators of nodes adjacent to $\mathrm{Min}(S)$ on $S$.
    We write $d(S) = \dist(S,\mathrm{Max}(S))$.
    Note that $\mathrm{Min}(S)$ only needs to collect its $4 \ell$-hop neighborhood in $G$ to determine $\mathrm{Max}(S)$. Each path or cycle $X$ satisfying $\dist(S,X) \leq 2\ell$ has to contain a node $v$ such that $\dist(\mathrm{Min}(S),v) \leq \ell + \dist(S,X) \leq 3 \ell$.
    Next, using just the $\ell$-hop neighborhood around the creator of $v$, $\mathrm{Min}(S)$ can determine whether $f(X) = \ell^\star(X)$ or $f(X) = \infty$.
    Further, if $\mathrm{Max}(S) \neq S$, $S$ identifies a target path or cycle $\mathrm{Tar}(S)$ such that $\dist(S,\mathrm{Tar}(S)) = 0$, that is, there are nodes $v \in S$ and $v' \in \mathrm{Tar}(S)$ that share a common creator, and $\dist(\mathrm{Tar}(S), \mathrm{Max}(S)) \leq d(S) - 1$.

    Now, for all short cycles $S \in \mathcal{S}$ such that $S \neq \mathrm{Max}(S)$
    in parallel, we modify the local vertex splits such that $S$ and $\mathrm{Tar}(S)$ are joined to form a larger path or cycle according to \autoref{lem:cycle_doubling}.
    After executing all those merge operations in parallel, let $\mathrm{New}(S)$ denote the new connected component that $S$ is now a part of. If $\mathrm{New}(S)$ is a path, we are done, as our goal is to delete short cycles. We do not need to argue about the true length of $\mathrm{New}(S)$ if it is a path. Let us assume in what follows that $\mathrm{New}(S)$ is a cycle, that is, $S$ only merged with other cycles, and let us analyze its length.

    We claim that the weak diameter of $\mathrm{New}(S)$ is at least $d(S)$, the distance between $S$ and $\mathrm{Max}(S)$. 
    Let $S = S_0,S_1,\dots,S_k$ denote the chain of cycle joins originating from $S$, where $S_{i+1} = \mathrm{Tar}(S_i)$ for all $i = 0,\dots,k-1$.
    First, we note that these merge chains can never close a cycle, that is, the graph induced by all cycle joins is acyclic itself.
    Indeed suppose that $\mathrm{Tar}(S_i) = S$ for some $i < k$.
    If $\mathrm{Max}(S_i) = \mathrm{Max}(S)$, then $S$ cannot be the target for $S_i$, as $\dist(S_i, \mathrm{Max}(S)) < \dist(S, \mathrm{Max}(S))$.
    On the other hand, if $\mathrm{Max}(S_i) \neq \mathrm{Max}(S)$, then it must hold that 
    \[
        f(\mathrm{Max}(S_i)) \geq f(\mathrm{Max}(S_{i-1})) \geq \cdots \geq f(\mathrm{Max}(S)).
    \]
    But $S$ is also the target of $S_i$, which contradicts the fact that either $f(\mathrm{Max}(S)) < f(\mathrm{Max}(S_i))$ or $\mathrm{Max}(S)$ loses in the tiebreak with $\mathrm{Max}(S_i)$.
    Thus, if we consider a set of cycles merging together, they form a rooted tree if we put a directed edge from $C$ to $C'$ in the set whenever $C' = \mathrm{Tar}(C)$.
    Executing the merge operations sequentially, from root to leaves, we progressively grow the root cycle with each merge until a single cycle remains, spanning the virtual nodes of all cycles in the rooted tree.

    Next, we argue that the weak diameter of $\mathrm{New}(S)$ is at least $d(S)$.
    This is clear, if $S_k = \mathrm{Max}(S)$.
    If $S_k \neq \mathrm{Max}(S)$, let $i^\star < k$ be the smallest $i$ such that $\mathrm{Max}(S_{i^\star}) \neq \mathrm{Max}(S)$. Since $\mathrm{Max}(S_{i^\star})$ beats $\mathrm{Max}(S)$ either in terms of true length or tiebreaks we must have that $\dist(\mathrm{Max}(S_{i^\star}), S) > \ell$.
    Since $f$ is monotonically increasing along the chain, this implies that the chain can never terminate at a cycle $C$ with $\dist(S,C) \leq \ell$.
    Hence, the weak diameter of $\mathrm{New}(S)$ is at least $\ell \geq d(S)$.

    If we are in the second case, we are already done, as the weak diameter is always a lower bound for the true length of a cycle.
    Now suppose there still exists a short cycle $S$, where the true length of $\mathrm{New}(S)$ is less than $\ell$. Unless the entire graph itself has small diameter (trivializing all our problems), there must be two cycles $C$ and $C_i$ such that $\dist(C,C_i) = 0, C_i \subseteq \mathrm{New}(S)$ but $C \not \subseteq \mathrm{New}(S)$, i.e., $C$ has not been merged with $S$. Since $\mathrm{Tar}(C) \neq C_i$ and $\mathrm{Tar}(C_i) \neq \mathrm{Tar}(C)$, we have that $\mathrm{Max}(C_i) = \mathrm{Max}(S)$. Since $\dist(\mathrm{Max}(S), C) \leq 2 \ell$, this implies that $\dist(S, \mathrm{Max}(C_i)) > 2 \ell$.
    Hence, after joining $\mathrm{New}(S)$ with $\mathrm{New}(C)$, we get a cycle with weak diameter at least $\ell$ (or a path, if $\mathrm{New}(C)$ is a path).
    Importantly, since $\dist(S,S_i) < \ell$, this implies that $\mathrm{New}(C)$ already had a weak diameter of at least $\ell$ on its own.
    Since this step only joins cycles of true length less than $\ell$ to cycles of weak diameter at least $\ell$ (or paths), we get that our graph of cycle joins remains acyclic and we never execute a cycle join on two virtual nodes that are part of the same cycle.
\end{proof}

\subsection{The Algorithm}

\SplittingThm*

\begin{proof}[Proof of \autoref{thm:balanced-local} (undirected)]
    First we choose $x(v) = \max \left\{1, \left\lfloor \frac{\varepsilon \deg(v)}{4} \right\rfloor \right\} $ and $\ell = \lceil 16 \varepsilon^{-1} \rceil $.
    After constructing $G_\mathrm{split}$ as described in \autoref{def:G_split} we apply \autoref{lem:increase_girth} to $G_\mathrm{split}$.
    This yields a modified split graph $G_\mathrm{split}'$, where the true length of each cycle is at least $\ell$.
    Hence, according to \autoref{lem:local_voting_block} there is at least one voting block $B \in \mathcal{P}$ participating in the HSO for every cycle.
    Hence, after chopping the segments according to the HSO assignment, each cycle is split into at least two paths.

    In order to reduce the discrepancy for odd-degree vertices down to just $1$, we slightly modify our HSO instance. 
    If a bucket $B$ contains a virtual node $v$ of degree one (created by one of the odd-degree vertices in $G$), then we are not allowed to chop at another node in $B$. 
    However, we cannot afford to let all buckets $B$ simply drop out of the HSO instance. 
    Therefore, we show that for every bucket $B$ already containing a degree-one node $v$, we can do an alternative operation that also chops a path into smaller segments.
    For that, we let the degree-two node $w$ that `won' the bucket $B$ exchange one of its edges with the only edge incident to $v$. 
    This is essentially a `one-sided' chop, as now one side of the path containing $w$ ends at $v$, while the other side now continues along the path $P_v$ that previously started at $v$.
    These one-sided chops are a little bit more precarious to work with, as now we still need to make sure that the path that now continues along $P_v$ eventually does terminate.
    For that, we distinguish two cases: 
    If $P_v$ receives fewer than two voting blocks, then we merge the two hyperedges corresponding to the buckets at either endpoint of $P_v$ into a single hyperedge.     
    Hence, at most one endpoint of $P_v$ exchanges its edge with another virtual node of degree $2$ and the path has to stop at the other endpoint of $P_v$.
    Note that the weak diameter of $P_v$ is still at most $\bigO(\ell)$ and the rank of $H$ increases by a factor of two, which can be compensated by setting $\ell = \lceil 16 \varepsilon^{-1} \rceil $.
    The runtime of the HSO algorithm therefore remains at $\bigO(\ell \cdot \log n)$.

    If a path $P$ receives at least two voting blocks and a voting block $B$ that is the closest voting block to one of the endpoints of $P$ has to do a one-sided chop, then it always chops in such a way that the subpath coming from the endpoint of $P$ immediately terminates at $B$.
    Thus, we can still guarantee that the weak diameter of each path and cycle in $G_\mathrm{split}$ is bounded by $\bigO(\ell)$.

    Now we can color the edges of all those paths red and blue such that each internal node of each path receives one red and one blue edge.
    This yields a discrepancy of at most $\varepsilon \cdot \deg(v) + 2$ for even-degree vertices and $\varepsilon \cdot \deg(v) + 1$ for odd-degree vertices.
    Removing all short cycles from $G_\mathrm{split}$ takes $\bigO(\ell)$ rounds, executing HSO still works in $\bigO(\ell \cdot \log n)$ rounds and the final coloring can be computed in $\bigO(\ell \cdot \log n)$ rounds.
    Since $\ell = \bigO(\varepsilon^{-1})$, this implies that our algorithm runs in $\bigO(\varepsilon^{-1} \cdot \log n)$ rounds in total.
\end{proof}

\SplittingThmRand*

\begin{proof}[Proof of \autoref{thm:balanced-local-rand} (undirected)]
    The randomized algorithm differs from its deterministic counterpart only in the subroutines it calls. The deterministic ruling set algorithm, as well as the deterministic HSO algorithm are replaced by their randomized counterparts.
    Hence, we replace \autoref{lem:ruling_set_voting_blocks} by \autoref{lem:ruling_set_voting_blocks_randomized} to reduce the runtime for constructing the hypergraph $H$ down to $\bigO(\ell \cdot \log \log n)$.
    Next, we choose $\ell = \lceil 16 \cdot 10^{\log^{\ast} \varepsilon^{-1} } \cdot \varepsilon^{-1} \rceil \geq 2 \cdot 10^{\log^{\ast} r} \cdot r$. 
    The construction of $H$ takes $\bigO(\ell \cdot \log \log n)$ rounds.
    Now we apply \autoref{thm:better-rand-HSO} to solve HSO in $\mathcal{O}(\log \log n)$
    rounds. 
    Since simulating a round on $H$ takes $\bigO(\ell)$ rounds on $G$ we can bound the overall runtime of our algorithm by $\bigO(\ell \cdot \log \log n) = \bigO(\varepsilon^{-1} \cdot 10^{\log^{\ast} \varepsilon^{-1}} \cdot \log \log n)$.
\end{proof}

\section{Applications}

\begin{corollary}
    Let $k$ be an integer and $\eps \in (0,2k/3]$.
    There is a deterministic distributed algorithm of round complexity $\bigO(k^2 \cdot \varepsilon^{-1} \cdot \log n)$ for splitting the edges of a graph $G$ into $2^k$ parts such that each vertex $v \in V$ is incident to at most $((1+ \varepsilon) / 2^k) \cdot \deg(v) + 6$ edges of any part.
\end{corollary}

\begin{proof}
    The statement is trivial for $k=0$ and covered by the undirected version of \autoref{thm:balanced-local} for $k=1$.
    Set $\varepsilon' = \varepsilon / (2 k)$. Note that $\eps' \leq 1/3$.
    We apply the undirected version of \autoref{thm:balanced-local} to $G$ and recurse on the two edge-induced subgraphs $G_\mathrm{red}$ and $G_\mathrm{blue}$. After $k$ iterations, we obtain $2^k$ subgraphs, where each vertex $v$ has degree at most
    \begin{align*}
    \deg_k(v)
    & =
    \parens*{\frac{1 + \varepsilon'}{2}}^k \deg(v)
    + 2 \sum_{\ell = 0}^{k - 1 }
    \parens*{\frac{1 + \varepsilon'}{2}}^\ell
    \\
    &
    \leq 2^{-k} \cdot e^{\eps' \cdot k} \cdot \deg(v)
    + 2 \sum_{\ell = 0}^{\infty}
    \parens*{\frac{1 + 1/3}{2}}^\ell
    = 2^{-k} \cdot \exp(\varepsilon / 2) \cdot \deg(v) + 6
    \\
    & \leq \frac{1+\varepsilon}{2^k}\cdot \deg(v) + 6.
    \end{align*}
    The runtime of each iteration is bounded by $\bigO((\varepsilon')^{-1} \cdot \log n) = \bigO(k \cdot \varepsilon^{-1} \cdot \log n)$ and thus over all $k$ iterations the overall runtime is bounded by $\bigO(k^2 \cdot \varepsilon^{-1} \cdot \log n)$.
\end{proof}

Remark that the $+6$ term in the discrepancy can be removed in graphs of minimum degree at least $12 \cdot 2^k / \epsilon$, by calling the corollary with $\eps/2$ instead of $\eps$.

\subsection{Edge Coloring}

\begin{lemma}
    For any deterministic edge coloring algorithm $\mathcal{A}$ using $k \cdot \Delta$ colors running in time $f(\Delta,n)$ and any $\varepsilon$ such that $k \leq \frac{\varepsilon}{2 (\exp(\varepsilon / 4) - 1)}$, there is a deterministic edge coloring algorithm $\mathcal{A}'$ using $(k + \varepsilon) \cdot \Delta$ colors running in time $\bigO(\varepsilon^{-1} \cdot \log^2 \Delta \cdot \log n) + f(\bigO(\varepsilon^{-1}), n) $.
\end{lemma}

\begin{proof}
    Let $G = (V,E)$ be an $n$-vertex graph of maximum degree $\Delta$.
    If $\varepsilon < 6/\Delta$, we simply apply algorithm $\mathcal{A}$ directly as $f(\Delta, n) = f(\mathcal{O}(\varepsilon^{-1}), n)$.
    Otherwise, set $\varepsilon' = \varepsilon / (4 \log \Delta)$ and let $h = \lfloor \log (\varepsilon \Delta / 6) \rfloor$.
    We apply \autoref{thm:balanced-local} to $G$ and obtain two subgraphs $G_\mathrm{red}$ and $G_\mathrm{blue}$ of maximum degree $\Delta_1 = (1+\varepsilon')/ 2 \cdot \Delta + 1$.
    Now we recurse on the resulting subgraphs in parallel for $h$ iterations.
    This procedure yields $2^h$ subgraphs, each of maximum degree
    \[
        \Delta_h = \left( \frac{1 + \varepsilon'}{2} \right)^h \cdot \Delta
        + \sum_{i = 0}^{h - 1} \left( \frac{1 + \varepsilon'}{2} \right)^i
        \leq 2^{-h} \cdot (1 +  \varepsilon' )^h \cdot \Delta + \frac{2}{1 - \varepsilon'}.
    \]
    Therefore we have that
    \[
        \Delta_h \leq 2^{-h} \cdot (1 + \varepsilon')^h \cdot \Delta +  \frac{2}{1 - \varepsilon'} 
        \leq 12 \varepsilon^{-1} \cdot \exp(h \cdot \varepsilon') + 3
        \leq 12 e \cdot \varepsilon^{-1} + 3.
    \]
    Now we execute algorithm $\mathcal{A}$ on all those $2^h$ subgraphs in parallel, using disjoint color palettes. The total number $k'$ of colors used is therefore
    \begin{align*}
        k' &= 2^h \cdot (k \cdot \Delta_h) \leq k \cdot \Delta \cdot (1 + \varepsilon / (4 \log \Delta))^h + 2^h \cdot 3
        \\
        &\leq k \cdot \Delta \cdot (1 +  \varepsilon / (4 \log \Delta))^{\log \Delta} + \frac{\varepsilon}{2} \Delta \\
        &\leq k \cdot \Delta \cdot \exp(\varepsilon / 4) + \frac{\varepsilon}{2} \Delta \\
        &\leq (k + \varepsilon) \cdot \Delta.
    \end{align*}
    Each of the $h = \bigO(\log \Delta)$ levels of the recursive split takes $\bigO((\varepsilon')^{-1} \cdot \log n) = \bigO( \varepsilon^{-1} \cdot \log \Delta \cdot \log n)$ rounds, while simulating algorithm $\mathcal{A}$ takes $f(\bigO(\varepsilon^{-1}),n)$ rounds.
\end{proof}

Plugging in the $(3\Delta / 2)$-edge coloring algorithm from \cite{BMNSU_soda25} running in $\bigO(\Delta^2 \cdot \log n)$ rounds in \LOCAL we get the following corollary.

\begin{corollary}
    For any $\varepsilon > 1 / \Delta$, there is a deterministic \LOCAL algorithm of complexity $\bigO(\varepsilon^{-1} \cdot \log^2 \Delta \cdot \log n + \varepsilon^{-2} \cdot \log n)$ that computes a $(3/2 + \varepsilon)\Delta$-edge coloring on any $n$-vertex graph with maximum degree $\Delta$.
\end{corollary}

\subsection{HSO in CONGEST}

As a warm-up we show that HSO can also be solved efficiently in \CONGEST albeit with a slightly relaxed degree-rank condition.
The warm-up is a simple example of how directed degree splitting can be used to solve some instances of HSO, and the results afterwards build on the same idea.

\begin{lemma}
\label{lem:hso-det-congest-powers}
    Let $k \in \mathbb{N}$ be a positive integer.
    Then, there is a $\bigO(k \cdot \log n)$-round \CONGEST algorithm that computes a hypergraph sinkless orientation in any $r$-uniform hypergraph with $r = 2^k$ and minimum degree $\delta = 3^k$.
\end{lemma}

\begin{proof}
    We execute $k - 1$ iterations of the following degree-rank reduction scheme: Given a $2^{i}$-uniform hypergraph $H$ with minimum degree $3^{i}$, we compute a $2^{i-1}$-uniform subhypergraph $H'$ with minimum degree $3^{i-1}$.
    Let $B = (V \sqcup E, F)$ be the bipartite representation of $H$.
    For each vertex $v \in V$, we divide the hyperedges incident to $v$ arbitrarily among $3^{i-1}$ copies of $v$ such that each copy has degree at least three. Similarly, we split each hyperedge $e \in E$ into $2^{i-1}$ copies of $e$ such that each copy contains at most two vertices in $V$.
    On this virtual graph, we compute a sinkless orientation in $\bigO(\log n)$ rounds using \autoref{lem:sink-sourceless}.
    Importantly, we also need that the degree-two copies of each hyperedge are neither sources nor sinks. This can be achieved by contracting the two edges incident to a single copy before running the sinkless orientation algorithm.
    This leads to an orientation of $B$, where each vertex $v \in V$ has at least $3^{i-1}$ outgoing edges. Further, each edge $e \in E$ has exactly $2^{i-1}$ incoming edges.
    Now let $H'$ be the subhypergraph obtained by removing all edges oriented from $E$ to $V$ in $B$.
    Then $H'$ has minimum degree $\delta' \geq 3^{i-1}$, while each edge has rank $r' = 2^{i-1}$.
    After $k - 1$ iterations of this procedure, we are left with a hypergraph $H^\star$ of minimum degree $3$, where every edge has rank $2$. 
    Thus, we have reduced our original instance $H$ to a graph $H^\star$, where each vertex in $H$ corresponds to exactly one vertex in $H^\star$ and each hyperedge in $H$ corresponds to exactly one edge in $H^\star$.
    Finally, we compute a sinkless orientation on $H^\star$, which immediately yields a hypergraph sinkless orientation of $H$.
    Since each of the $k$ iterations of the degree-rank reduction takes $\bigO(\log n)$ rounds in \CONGEST, the overall procedure takes $\bigO(k \cdot \log n)$ rounds.
\end{proof}

\begin{corollary}
\label{cor:hso-det-congest-basic}
    There is a $\bigO(\log r \cdot \log n)$-round deterministic \CONGEST algorithm that computes a hypergraph sinkless orientation in any $r$-uniform hypergraph with minimum degree $\delta \geq 3\cdot r^{\log_2(3)}$.
\end{corollary}
\begin{proof}
    Let $k = \ceil{\log_2(r)}$. The hypergraph has maximum rank $r \leq 2^k$, and minimum degree $\delta \geq 3\cdot r^{\log_2(3)} > 3\cdot (2^{k-1})^{\log_2(3)} = 3^k$.
    This allows us to apply \autoref{lem:hso-det-congest-powers}.
\end{proof}

In order to recover a tighter degree-rank condition we use a more balanced degree-rank reduction inspired by \cite{weakSplitting19}.

\begin{algorithm*}[ht!]
	\caption{Degree-rank reduction}
    \label{alg:degree_rank_reduction}
	\begin{algorithmic}[1]
        \Statex \textbf{Input:} A bipartite graph $B = (U \sqcup V,E), r := \max_{v \in V} \deg(v), \delta := \min_{u \in U} \deg(u)$ and a parameter $0 < \varepsilon \leq 1/3$.
        \Statex \textbf{Output:} A subgraph $B'$ with $r' \leq (1/2 + \varepsilon)r + 2$ and $\delta' \geq (1/2 - \varepsilon)\delta - 2$.
        \State Compute a directed split on $B$ with discrepancy at most $\varepsilon \cdot \deg(v)$ for all $v \in U \cup V$.
        \State Let $B'$ be the graph obtained by deleting all edges oriented from $V$ to $U$ in $B$.
	\end{algorithmic}
\end{algorithm*}

\begin{lemma}
    \label{lem:HSO_CONGEST}
    Let $T_\mathrm{split}(n,\varepsilon)$ denote the runtime of a directed degree splitting algorithm with discrepancy at most $\varepsilon \cdot \deg(v)$ in \CONGEST.
    Then, there are distributed algorithms that compute a hypergraph sinkless orientation in any hypergraph with maximum rank $r$ and minimum degree $\delta \geq 50 r$ in
    \begin{itemize}
        \item $\bigO(T_\mathrm{split}(n, 1 / \log r) \cdot \log r + \log n)$ rounds of deterministic \CONGEST and
        \item $\bigO(T_\mathrm{split}(n, 1 / \log r) \cdot \log r + \log \log n)$ rounds of randomized \CONGEST.
    \end{itemize}
\end{lemma}

\begin{proof}
    We execute $k = \lceil \log r \rceil$ iterations of \autoref{alg:degree_rank_reduction} for $\varepsilon = 1/(10k)$.
    Let $\delta_k$ (and $r_k$) denote the minimum degree (and the maximum rank) of the hypergraph $H_{k}$ obtained after $k$ iterations.
    Via induction over the number of iterations one can show that
    \[
        \delta_k \geq \left( \frac{1 - \varepsilon}{2}\right)^k \delta - 2\sum_{i=0}^{k-1} \left(\frac{1 - \varepsilon}{2}\right)^{i} > \left( \frac{1 - \varepsilon}{2} \right)^k \delta - 2 \sum_{i= 0}^{\infty}  \frac{1}{2^i} = \left( \frac{1 - \varepsilon}{2} \right)^k \delta - 4,
    \]
    as well as
    \[
        r_k \leq \left( \frac{1 + \varepsilon}{2}\right)^{k}r + 2\sum_{i=0}^{k-1} \left( \frac{1 + \varepsilon}{2}\right)^{i} \leq \left( \frac{1 + \varepsilon}{2}\right)^{k}r + \frac{2}{1 - \frac{1+\varepsilon}{2}} \leq \left( \frac{1 + \varepsilon}{2}\right)^{k}r + 6.
    \]
    Plugging in our chosen values for $k$ and $\varepsilon$ we get that
    \[
       \delta_{k} > \left( \frac{1-\varepsilon}{2} \right)^{k}\delta - 4 = \left( \frac{1-\frac{1}{10k}}{2} \right)^{k}\delta - 4 > \frac{9}{10} \cdot \frac{1}{2^k}\delta - 4 > \frac{9}{10} \cdot \frac{\delta}{2r} - 4 \geq 15.
    \]
    as well as
    \[
        r_{k} < \left( \frac{1+\varepsilon}{2} \right)^{k}r + \frac{40}{9} = \left( \frac{1+\frac{1}{10k}}{2} \right)^{k}r + \frac{40}{9} < \frac{6}{5} \cdot \frac{1}{2^k}r + \frac{40}{9} < \frac{6}{5} + \frac{40}{9} < 6.
    \]
    Now observe that two more iterations of \autoref{alg:degree_rank_reduction} reduces the minimum degree from at least $15$ to at least $7$ to at least $3$.
    The maximum rank on the other hand reduces from at most $5$ to at most $3$ to at most $2$.
    Thus, the hypergraph $H_{k+2}$ is indeed a graph (that might contain loops) of minimum degree $3$.
    Hence, using \cite{NolinBA} we can compute a sinkless orientation on $G := H_{k+2}$ in $\mathcal{O}(\log n)$ rounds of deterministic \CONGEST or $\mathcal{O}(\log \log n)$ rounds of randomized \CONGEST.
    Finally note that a sinkless orientation on $G$ immediately induces a hypergraph sinkless orientation on $H$: Each vertex in $H$ also appears as a vertex in $G$ and thus receives at least one outgoing hyperedge. On the other hand, each hyperedge in $H$ corresponds to at most one edge in $G$, which is marked outgoing for exactly one incident vertex. For every vertex $v$ that dropped out of a hyperedge $e$ at any point of the degree-reduction procedure, we mark $e$ as incoming for $v$.
\end{proof}

Although not explicitly stated in the paper, the degree splitting algorithm from \cite{GHKMSU_dc20} also works in the \CONGEST model. 
Plugging in their runtime into \autoref{lem:HSO_CONGEST} we get the following concrete runtimes:

\begin{corollary}
    \label{cor:HSO_CONGEST}
    There is a distributed algorithm for computing a hypergraph sinkless orientation in any hypergraph with maximum rank $r$ and minimum degree $\delta \geq 50r$ in $\bigO(\log^2 r \cdot \log \log r \cdot (\log \log \log r)^{1.71} \cdot \log n)$ rounds of deterministic \CONGEST.
\end{corollary}

\begin{corollary}
\label{cor:rand-hso-from-splitting}
    There are distributed algorithms for computing a hypergraph sinkless orientation in any hypergraph with maximum rank $r$ and minimum degree $\delta \geq 50r$ in 
    \begin{itemize}
        \item $\bigO(\log^2 r \cdot 10^{\log^\ast r} \cdot  \log \log n)$ rounds of randomized \LOCAL and
        \item $\bigO(\log^2 r \cdot \log \log r \cdot (\log \log \log r)^{1.71} \cdot \log \log n)$ rounds of randomized \CONGEST. 
    \end{itemize}
    Note that for $r = \omega(\log n)$, the dependency on $r$ can be reduced to $\log n$ by the subsampling technique presented in the proof of \autoref{thm:balanced-local}.
\end{corollary}

To sum up our current understanding of hypergraph sinkless orientation in randomized \LOCAL, we know the problem to be solvable in the following complexities depending on the ratio between the parameters $\delta$ and $r$:
\begin{itemize}
    \item when only assuming $\delta/r > 1$, we know the problem to be solvable in $\bigO(\log_{\delta/r} n)$ by using the deterministic algorithm for the problem (\autoref{thm:hso}),
    \item assuming a larger ratio of $\delta / r\geq 50$, we know the problem to be solvable in $\bigO(\log^2 r \cdot \log \log n)$ rounds, and in fact, in $\bigO(\log^3 \log n)$ by using a simple subsampling argument (\autoref{cor:rand-hso-from-splitting}),
    \item for an even larger ratio of $\delta / r \geq 2x \cdot 10^{\log^* r}$ for any value $x\geq 2$, we have an algorithm of complexity $\bigO(\log_x \log n)$
    (\autoref{thm:better-rand-HSO}). The complexity is even $\bigO(1)$ when $\delta / r = \Omega(\log n)$.
\end{itemize}
Note that an algorithm of complexity $\bigO(\log_{\delta/r} \log n)$ would subsume all these results, leaving as a natural open question whether such a complexity is achievable. This complexity is a natural target, as it corresponds to the complexity of deterministically solving similar HSO instances of size $\poly \log n$. 

\section{Mending Radius Lower Bound}

Before defining the mending radius, let us formally define what we mean by a partial solution to a problem.
For simplicity, we only give the definition of a partial solution for directed splitting. For the more general definition that applies to all locally checkable labeling problems, see~\cite{BHMORS_sirocco22}.

For a graph $G=(V,E)$, we represent a partial orientation of the edges of that graph as an assignment $\psi:E \to \set{-,+,\bot}$.
For an edge $e \in E$, $\psi(e) = \bot$ corresponds to $e$ being unoriented.
For an edge $e=uv \in E$, with $u < v$ (i.e., $v$ has a larger identifier than $u$), $\psi(uv) = +$ represents that the edge is oriented from $u$ to $v$. Conversely $\psi(uv) = -$ indicates that the edge is oriented in the opposite way, from $v$ to $u$.
When $\psi(e) \neq \bot$ for all $e\in E$, $\psi$ is said to be a full orientation.
We say that a partial solution $\psi'$ extends another partial solution $\psi$ if $\psi'^{-1}(\bot) \subseteq \psi^{-1}(\bot)$ and for all $e$ with $\psi(e) \neq \bot$ it holds that $\psi(e) = \psi'(e)$.
For two solutions $\psi$ and $\psi'$, their difference $\psi \triangle \psi'$ is the set of edges $\set{e \in E \mid \psi(e) \neq \psi'(e)}$.

For each vertex $v$, let us partition its neighbors according to whether their identifiers are lower or greater than that of $v$. That is, we define $N^+(v) = \set{u \in N(v) \mid u > v}$ and $N^-(v) = N(v) \setminus N^+(v)$.
A partial solution to the problem of computing a directed splitting with discrepancy $\eps\Delta$ is an assignment $\psi$ satisfying the following properties at each vertex $v$
\begin{align*}
\card{\psi^{-1}(+) \cap \{ uv \in E: u \in N^+(v)\}} + \card{\psi^{-1}(-) \cap \{ uv \in E: u \in N^-(v)\}} \leq \frac{1+\eps}{2}\Delta\ ,
\\
\card{\psi^{-1}(-) \cap \{ uv \in E: u \in N^+(v)\}} + \card{\psi^{-1}(+) \cap \{ uv \in E: u \in N^-(v)\}} \leq \frac{1+\eps}{2}\Delta\ .
\end{align*}
Intuitively, the first equation indicates that a vertex $v$ does not have too many outgoing edges, while the second indicates that it does not have too many incoming edges.
When these equations hold at a vertex $v$, it is always possible to give an orientation to the unoriented edges incident to $v$ such that the discrepancy at $v$ is at most $\eps \Delta$.

However this is a very local property. It might not be possible to give an orientation to unoriented edges such that all vertices have discrepancy below $\eps \Delta$.
In fact, we will show that there exists a graph $G=(V,E)$, a partial assignment $\psi$, and an edge $e \in \psi^{-1}(\bot)$, such that for any partial assignment $\psi'$ extending $\psi$ on $e$ (i.e., $\psi'(e) \neq \bot$), the graph induced by the edges in $\psi \triangle \psi'$ has diameter $\Omega(\eps^{-1}\log n)$.
The mending radius measures exactly this: given a partial assignment $\psi$, which we extend to a new edge $e$, how large can the diameter of the subgraph induced by the changes made to $\psi$ be?

\begin{definition}[Mendability, mending radius~\cite{BHMORS_sirocco22}]
    For any integer $T$, a partial solution $\psi$ on a graph $G$ is $T$-mendable at a vertex $v$ if there exists a valid partial solution $\psi'$ extending $\psi$ s.t.\ $\psi'(e) \neq \bot$ for all edges $e$ incident to $v$ and for every edge $uu' \in \psi \triangle \psi'$, $\max_{w\in\set{u,u'}} \dist_G(v,w) \leq T$.
    
    A problem $\Pi$ on input graph family $\mathcal{G}$ is $T$-mendable for a non-decreasing function $T:\mathbb{N} \to \mathbb{N}$ if for every graph $G =(V,E) \in \mathcal{G}$, every valid partial solution $\psi$ on $G$, and every vertex $v$, $\psi$ is $T(\card{V})$-mendable at $v$.

    A problem $\Pi$ on input graph family $\mathcal{G}$ has mending radius $\Theta(T(n))$ for a non-decreasing function $T:\mathbb{N} \to \mathbb{N}$ if there exist two functions $T',T''$ s.t.\ $T'(n) \in \Theta(T(n))$, $T''(n) \in \Theta(T(n))$, and $\Pi$ on input graph family $\mathcal{G}$ is $T'$-mendable but not $T''$-mendable.
\end{definition}

To show our mending lower bound, at a high level, we build a graph consisting of $\Omega(\eps^{-1} \log n)$ ordered layers. The lowest layer contains unoriented edges between its vertices, while all other edges connect vertices in different layers and are oriented towards lower layers.
Each layer is only connected to the two layers preceding it and the two layers succeeding it, so that the graph has diameter $\Omega(\eps^{-1} \log n)$. Having all edges pointing in the direction of the lower layers makes it impossible to give an orientation to edges in the lowest layer without making changes in the higher layers.
A similar idea was used in \cite{CHLPU_talg20} in the context of sub-$4\Delta/3$-edge coloring.

\thmMendability*
    
\begin{proof}[Proof of the directed case]
    For simplicity throughout the proof, we assume $\eps \Delta/2$ and $\Delta/2$ to be integers, so that a fortiori $(1-\eps)\Delta/2$ and $(1+\eps)\Delta/2$ are integers.
    
    We build a graph $G=(V,E)$ and a partial solution $\psi$ on $G$
    that is not $c\eps^{-1} \log (n/\Delta)$-mendable at some vertex $v$ for a small enough constant $c\in \Theta(1)$.
    At a high level, we construct a graph and a solution in which a lot of edges are oriented towards some unoriented clique.
    In oriented parts of the graph, all vertices have the maximum allowed discrepancy, so as to create an imbalance between the number of edges pointing towards the unoriented clique and the number of edges pointing away from it.
    In the unoriented clique, the unoriented edges ensure that no vertex has a discrepancy above the allowed threshold, but extending the solution to the clique cannot be done without changing the orientation of some edges at distance $\Omega(\eps^{-1} \log (n/\Delta))$ from the clique.
    
    The graph consists of $D \in \Theta(\eps^{-1} \log (n/\Delta))$ sets of vertices $V_1,\dots,V_D$, forming layers.
    For each $i\in [D-1]$, $\card{V_{i+1}} \approx \frac{1+\eps}{1-\eps}\card{V_i}$.
    Also, for each $i\in [3,D-2]$, each vertex $v \in V_i$ has $\frac{1+\eps}{2}\Delta$ edges to vertices in $V_{i+1} \sqcup V_{i+2}$, and $\frac{1-\eps}{2}\Delta$ to vertices in $V_{i-1} \sqcup V_{i-2}$.
    Vertices from the first set $V_1$ form a $(\frac{1-\eps}{2}\Delta+1)$-clique.
    The next layers $V_2,\dots,V_D$ are constructed iteratively as follows:
    \begin{itemize}
        \item Each vertex in $V_1$ has $\frac{1+\eps}{2}\Delta$ edges to vertices in $V_2 \sqcup V_3$, for a total of $\frac{1+\eps}{2}\Delta\card{V_1}$. Let $V_2$ be of size $\floor{\frac{1+\eps}{1-\eps}\card{V_1}}$, and add $\frac{1-\eps}{2}\Delta\card{V_2}$ edges between $V_2$ and $V_1$, respecting that each vertex in $V_2$ is connected to exactly $\frac{1-\eps}{2}\Delta$ vertices in $V_1$, and each vertex in $V_1$ is connected to at most $\frac{1+\eps}{2}\Delta$ vertices in $V_2$.
        Vertices in $V_1$ are missing $\frac{1+\eps}{2}\Delta\card{V_1} - \frac{1-\eps}{2}\Delta \floor*{\frac{1+\eps}{1-\eps}\card{V_1}}$ incident edges for all of them to be of degree $\Delta$. Let us denote this number of missing edges by $m_1$.
        By definition $m_1$ is between $0$ and $\frac{1-\eps}{2}\Delta$. We choose edges between $V_1$ and $V_2$ s.t.\ exactly one vertex $w_1 \in V_1$ is missing $m_1$ edges to be of degree $\Delta$.
        \item When constructing $V_3$, we connect $m_1$ vertices in $V_3$ to $w_1$, before adding edges between $V_2$ and $V_3$.
        $V_3$ is chosen to be of size $\floor{\frac{2}{(1-\eps)\Delta}\cdot(\frac{1+\eps}{2}\Delta\card{V_2}+m_1)}$.
        Again, we add edges between vertices of $V_2$ and $V_3$, such that all vertices in $V_3$ have exactly $\frac{1-\eps}{2}\Delta$ edges to $V_2 \cup \set{w_1}$, and all vertices in $V_2$ except one have $\frac{1+\eps}{2}\Delta$ edges to $V_3$. One vertex $w_2 \in V_2$ has only $\frac{1+\eps}{2}\Delta - m_2$ edges to $V_3$, where $m_2 \in [0 , \frac{1-\eps}{2}\Delta)$.
        \item We continue this construction recursively, choosing $\card{V_{i+2}} = \floor{\frac{2}{(1-\eps)\Delta}\cdot(\frac{1+\eps}{2}\Delta\card{V_{i+1}}+m_{i})}$ for each $i \in [1,D-2]$, defining a vertex $w_i$ and a number of missing edges $m_i$ for each layer.  
    \end{itemize}
    For $\eps$ close to $0$, layers grow by a multiplicative factor of about $1+2\eps$.
    As the first layer is of size $\Theta(\Delta)$, the process constructs a graph of size $\leq n$ when choosing $D$, the number of layers, to be $c\cdot \eps^{-1} \log (n/\Delta)$ with a small enough constant $c$.

    We now give a partial orientation $\psi$ of that graph, and argue that it is not $(D/2-1)$-mendable.
    For each edge between two vertices $u\in V_i$ and $v\in V_j$ s.t.\ $i<j$, we orient the edge towards $u$. That is, we orient edges towards lower layers. We do not orient edges in the clique formed by $V_1$.

    Consider an orientation $\psi'$ extending $\psi$ to edges of the clique formed by $V_1$. Let us denote $\bias(v,\psi) = \indeg(v) - \outdeg(v)$ the difference between the indegree and outdegree of a vertex in a partial orientation $\psi$.
    From all edges between $V_1$ and $V_2 \sqcup V_3$ being oriented towards $V_1$, we have that
    \begin{align*}
    \sum_{v \in V_1} \bias(v,\psi)
    & =
    \frac{1+\eps}{2}\Delta\card{V_1}\ .
    \intertext{
    Similarly, if we do the same sum over all layers from $V_1$ to $V_{D-2}$, we get
    }
    \sum_{i=1}^{D-2}\sum_{v \in V_i} \bias(v,\psi)
    &=
    \frac{1+\eps}{2}\Delta\card{V_1} + \eps\Delta\sum_{i=2}^{D-2}\card{V_i}
    >
    \eps\Delta\sum_{i=1}^{D-2}\card{V_i}.
    \end{align*}
    That is, the average bias of a node $v \in \bigcup_{i=1}^{D-2} V_i$ is strictly greater than $\eps \Delta$.

    Suppose that $\psi'$ only differs from $\psi$ on edges incident to the vertices in $\bigcup_{i=1}^{D-4} V_i$. Any change in that set of edges does not change the sum of biases.
    Indeed, when changing the orientation of an edge $uv$ in that set, the effect on the bias of $v$ is the opposite of the effect on the bias of $u$, and both biases are in the sum. Therefore, if $\psi'$ only differs from $\psi$ on edges within distance $D/2-2$ from $V_1$, there must exist a vertex of bias greater than $\eps \Delta$, meaning that the orientation does not have discrepancy $\leq \eps \Delta$ on all vertices.
    Mending the partial solution therefore requires changes at distance $\Omega(D) = \Omega(\eps^{-1} \log (n/\Delta))$ from where we want to extend the solution, and the problem has mending radius $\Omega(\eps^{-1} \log (n/\Delta))$.
\end{proof}

\begin{proof}[Proof of the undirected case]
    We build a different graph which again consists of $D \in \Theta(\eps^{-1}\log(n/\Delta))$ layers.
    The first two layers, with sets of nodes $V_1$ and $V_2$, form a bipartite graph where each node $v \in V_1$ is connected to $\Delta$ nodes in $V_2$, and each node in $V_2$ has $(1-\eps)\Delta/2$ neighbors in $V_1$. We give them size $\card{V_1} = (1-\eps)\Delta/2$ and $\card{V_2} = \Delta$.

    For each $i \in [D]$, the neighbors of a node in $V_i$ are contained in $V_{i-3}\sqcup V_{i-1} \sqcup V_{i+1}\sqcup V_{i+3}$ (with $V_i = \emptyset$ when $i<1$ or $i>D$).
    For each $i\in [2,D-3]$, as in the directed case, we have that $\card{V_{i+1}} \approx \frac{1+\eps}{1-\eps}\card{V_i}$. That is, the layers grow geometrically in size by a small $1+\bigO(\eps)$ factor. This is achieved, similar to the directed case, by having each node in $V_i$ have $(1+\eps)\Delta/2$ edges with nodes in $V_{i+1} \sqcup V_{i+3}$, and $(1-\eps)\Delta/2$ edges with nodes in $V_{i-1} \sqcup V_{i-3}$. Most of the edges are between successive layers, but some edges between layers $i$ and $j$ where $\abs{i-j} = 3$ are introduced to deal with ``rounding errors'' stemming from the fact that layers might not be able to grow by exactly a $\frac{1+\eps}{1-\eps}$ factor.
    
    We now describe a partial undirected splitting $\psi$ that is not $o(D)$-mendable. For all $i \in [2,D]$ such that $i = 0 \mod 2$, all edges between a node in $V_i$ and a node in $V_{i+1}\sqcup V_{i+3}$ are colored red.
    For all $i \in [3,D]$ such that $i = 1 \mod 2$, all edges between a node in $V_i$ and a node in $V_{i+1}\sqcup V_{i+3}$ are colored blue. Edges with an endpoint in $V_1$ are uncolored, and they are the only uncolored edges.

    For any $i \in [D]$ and a node $v \in V_i$, let $\bias(v,\psi) = (-1)^i(\reddeg(v) - \bluedeg(v))$ be the difference between the numbers of red and blue edges incident on $v$, with an alternating sign depending on the parity of $v$'s layer.

    If we sum the biases over all nodes in the first $D-3$ layers, we get 
    \begin{align*}
    \sum_{i=1}^{D-3}\sum_{v \in V_i} \bias(v,\psi)
    =
    \frac{1+\eps}{2}\Delta\card{V_1} + \eps\Delta\sum_{i=2}^{D-3}\card{V_i}
    >
    \eps\Delta\sum_{i=1}^{D-3}\card{V_i}.
    \end{align*}
    Changing the color of an edge between two layers $i,j \leq D-3$ does not change this sum, since changing the color of an edge $uv$ has opposite effects on the biases of $u$ and $v$. Therefore, unless we change the color of an edge incident to a node in a layer $V_i$ with $i \geq D-2$, there necessarily exists a node of bias $> \eps \Delta$, i.e., a node which does not satisfy the splitting's requirements.
    Extending the undirected splitting to the edges incident on $V_1$ thus requires changing the colors of some edges incident on nodes in layers $\geq D-2$, at distance $D/3-1$ from where the solution is being extended.
\end{proof}

An important caveat to our lower bound is that some problems have a mending radius far exceeding their complexity.
For instance, there exists a problem of mending radius $\Theta(n)$ but complexity $\bigO(\log^* n)$: coloring paths with $5$ colors, where the colors are partitioned into a set of size $2$ and another one of size $3$, and colors are only compatible within their set~\cite[Theorem~7.1, arXiv version]{BHMORS_sirocco22}. An even more pathological example is the following LCL with $3$ labels: $1$ and $2$ can only be adjacent to each other, and $3$ can only be adjacent to itself. The problem has a trivial $0$ round solution (output $3$ everywhere) but mending radius $\Omega(n)$.
These examples show what to look out for: the mending radius can be overly pessimistic, because a problem can admit fragile hard-to-mend partial solutions that an efficient algorithm will simply never consider.
In fact, many problems admit randomized algorithms of complexity lower than the mending radius, showing that randomness is one way to avoid bad partial solutions.

\appendix

\section{Complementary subroutines / lemmas}

\begin{lemma}[The Shattering Lemma {\cite[Lemma 6 (simplified)]{FG17}}]
    \label{lem:shattering}
    Let $G = (V,E)$ be a graph with maximum degree $\Delta$ and $c_{1},c_{2} \geq 1$ constants satisfying $c_{1} > 4c_{2} + 2$. Consider a process which generates a random subset $B \subseteq V$ such that $\mathbb{P}[v \in B] \leq \Delta^{-c_{1}}$ and such that the random variables $\mathbbm{1}[v \in B]$ depend only on the randomness of nodes within at most $c_{2}$ hops from $v$, for all $v \in V$. Moreover, let $H = G^{[2c_{2}+1,4c_{2}+2]}$ be the graph which contains an edge between $u$ and $v$ iff their distance in $G$ is between $2c_{2} + 1$ and $4c_{2}+2$. Then, with probability at least $1 - n^{-(c_{1} - 4c_{2} - 2)}$, the following holds:
    \begin{itemize}
        \item Any connected component in $G[B]$ has size at most $\bigO(\log_{\Delta}n \cdot \Delta^{2c_{2}})$.
    \end{itemize}
\end{lemma}

\begin{lemma}[Chernoff bound]
    \label{lem:chernoff}
    Let $X_1,\dots,X_n \in \{ 0, 1 \} $ be independent Bernoulli random variables and let $X = \sum_{i = 1}^{n} X_i$. Then, for any $\varepsilon > 0$ it holds that
    \[
        \mathbb{P}[X \geq (1 + \varepsilon) \cdot \mathbb{E}[X]] \leq \exp\left(- \frac{\varepsilon^2 \cdot \mathbb{E}[X]}{2 + \varepsilon}\right)
    \]
    and for any $\varepsilon \in [0,1]$ it holds that
    \[
        \mathbb{P}[X \leq (1 - \varepsilon) \cdot \mathbb{E}[X]] \leq \exp\left(\frac{-\varepsilon^2 \cdot \mathbb{E}[X]}{2}\right).
    \]
\end{lemma}

\begin{theorem}[{\cite[Theorem 8]{GV_RulingSetBoundedGrowth_podc07}}]
    \label{thm:ruling_set_induced_degree}
    There is a randomized distributed algorithm that computes a $(1,\bigO(\log\log\Delta))$-ruling set with induced degree at most $\bigO(\log^{5} n)$ in any $n$-vertex graph with maximum degree $\Delta$ in $\bigO(\log\log\Delta)$ rounds.
\end{theorem}

\begin{theorem}[{\cite[Theorem 3]{SEW_RulingSetColoring_tcs13}}]
    \label{thm:ruling_set_coloring}
    Let $G = (V,E)$ and $W \subseteq V$. Given a $d$-coloring of $G$, there is a deterministic distributed algorithm that computes a $(2,c)$-ruling set for $W$ in time $\bigO(c \cdot d^{1/c})$.
\end{theorem}

\begin{corollary}[randomized ruling set]
    \label{cor:rand_ruling_set}
    There is a randomized distributed algorithm that computes a $(2,\bigO(\log \log n))$-ruling set in any graph $G$ in $\bigO(\log \log n)$ rounds of the \LOCAL model.
\end{corollary}

\begin{proof}
    We first apply \autoref{thm:ruling_set_induced_degree} to get a $(1,\bigO(\log \log \Delta))$-ruling set $S_0$ with induced degree at most $\bigO(\log^5 n)$. 
    Next we compute a $\bigO(\log^{10} n)$-coloring $\chi$ of $G[S_0]$ using Linial's algorithm \cite{Linial1, linial92}. 
    Finally using $\chi$ as the input coloring and setting $c = \bigO(\log \log n)$, we get a $(2,\bigO(\log \log n))$ ruling set $S_1$ for $S_0$ in 
    $$
    \bigO(c \cdot d^{1/c}) = \bigO(\log \log n \cdot \log^{10 / \log \log n} n) = \bigO(\log \log n)
    $$
    randomized time. 
    Since every vertex $v \in V$ is at most $\bigO(\log \log \Delta) \subseteq \bigO(\log \log n)$ hops away from a vertex in $S_0$ and every vertex in $S_0$ is at most $\bigO(\log \log n)$ hops away from a vertex in $S_1$, we get that $S_1$ is also a $(2,\bigO(\log \log n))$-ruling set for $V$. 
\end{proof}

\bibliographystyle{plainurl}
\bibliography{0a_refs_preamble.bib, 0b_refs_sinksparse.bib}

\end{document}